\documentclass[12pt]{iopart}

\expandafter\let\csname equation*\endcsname\relax
\expandafter\let\csname endequation*\endcsname\relax
\usepackage{iopams}

\newtheorem{definition}{Definition}

\usepackage{graphicx}
\usepackage{xcolor}
\usepackage{tikz}
\usepackage{tikz-cd}
\usepackage{caption}
\usepackage{adjustbox}

\newcommand{\be}{\begin{equation}}
\newcommand{\ee}{\end{equation}}

\newcommand{\bff}{\boldsymbol}

\newtheorem{thm}{Theorem}[section]
\newtheorem{prop}[thm]{Proposition}

\newtheorem{lemma}[thm]{Lemma}

\newtheorem{remark}[thm]{\it Remark}
\newtheorem{example}[thm]{\it Example}

\newenvironment{proof}{\paragraph{Proof:}}{\hfill$\square$}

\usepackage{tikz-3dplot}

\tdplotsetmaincoords{60}{125}
\tdplotsetrotatedcoords{0}{20}{0} 

\newlength{\perspective}

\usepackage[
labelfont=sf,
hypcap=false,
format=hang
]{caption}

\usepackage{iopams}

\usepackage[xcolor]{changebar}
\usepackage{lipsum}

\usepackage{mathtools}

\begin{document}

\title[]{Non-Abelian hierarchies of compatible maps, associated integrable difference systems and Yang-Baxter maps}

\author{Pavlos Kassotakis}
\address{Pavlos Kassotakis, Department of Mathematical Methods in Physics, Faculty of Physics,
University of Warsaw, Pasteura 5, 02-093, Warsaw, Poland}
\ead{Pavlos.Kassotakis@fuw.edu.pl, pavlos1978@gmail.com}

\date{\today}


\begin{abstract}
We present two non-equivalent families of  hierarchies of non-Abelian compatible maps and we provide their Lax pair formulation.
These maps are associated with families of  hierarchies of non-Abelian Yang-Baxter maps, which we provide explicitly.  In addition, these hierarchies correspond to integrable difference systems with variables defined on edges of an elementary cell of the $\mathbb{Z}^2$ graph, that in turn lead to
 hierarchies of difference systems with variables defined on  vertices of the same cell. In that respect we obtain  the non-Abelian  lattice-modified Gel'fand-Dikii hierarchy, together with the explicit form of a non-Abelian hierarchy that we refer to as the lattice-NQC (or lattice-$(Q3)_0$) Gel'fand-Dikii hierarchy.
\end{abstract}

\ams{37K60, 39A14, 37K10, 16T25}
\vspace{2pc}
\noindent{\it Keywords}: $3D-$compatibe maps,  non-Abelian integrable difference systems, Yang-Baxter maps




\section{Introduction}
Unlike the scalar case, contributions on  integrable multi-component and integrable difference systems defined on higher order stencils, are rather sparse \cite{Nijhoff:1996,Korepanov:1997,Tongas:2004,Maruno:2010,Hietarinta:2011,JamesPhd,Hay:2013,Hay:2014,Mikhailov:2016,
Nalini:2018,Kels:2019,Kels:2019II,Kass2,Kamp:2020,Hietarinta:2020,Franks-book}.   Furthermore, unlike the continuous setting, contributions on non-Abelian generalizations and extensions of integrable difference systems are rather scattered in the literature \cite{Nijhoff:1990,Boris:2000,Dimakis:2002,bs:2002N,Field:2005,Nimmo:2006} and moreover quite rare when they concern hierarchies \cite{Doliwa_2013,Doliwa_2014,Noumi:2020,Kassotakis:1:2021,Dimakis_Korepanov:2021}.

The results of this paper serve in the renewed and growing interest in deriving and extending integrable difference systems and structures to the
non-Abelian domain \cite{Doliwa_2013,Doliwa_Painleve_2013,Grahovski:2016,Rizos:20182,Kass1,Papamikos:2021,Kassotakis:2:2021,Doliwa:2022,Doliwa2:2022}.
 The term non-Abelian, refers to the requirement that the multiplication operator is no longer Abelian. In that respect, the variables that participate on  difference systems do not a-priori mutually commute.

Specifically, in this article we introduce non-Abelian hierarchies of integrable difference systems in edge and in vertex variables. In detail,
  we derive two  $3D-$compatible non-Abelian hierarchies  of maps that we refer to as the  $\mathcal{K}^{(i)},$ $i=0,1$ hierarchies.  We  prove their multidimensional compatibility  and we  implicitly provide the companion maps that constitute hierarchies of Yang-Baxter maps. These companion hierarchies are the ones associated with the so-called $\mathcal{K}_{I},$ and $\Lambda_I$  Yang-Baxter maps introduced in \cite{Kassotakis:2:2021}. In the Abelian setting,  $\mathcal{K}_{I}$ and $\Lambda_I$ maps are equivalent to the so-called Harrison map\footnote{The Harrison map was  derived in \cite{pap2-2006} and serves as the nonlinear superposition formula for the B\"acklund transformation of the Ernst equation \cite{Harrison:1978} in general relativity}  a.k.a $H_I$ in \cite{Papageorgiou:2010}, so in that respect and in the Abelian restriction, our results provide the hierarchy of the $H_I$ Yang-Baxter map. In addition, we  show that both  $\mathcal{K}^{(i)}$ hierarchies  arise through Lax pair formulation by  deforming known Lax matrices.

Moreover,  these hierarchies serve as  deformations of the edge-variable avatars of the non-Abelian lattice Gel'fand-Dikii hierarchies introduced in \cite{Doliwa_2013}.  In the Abelian  case lattice Gel'fand-Dikii hierarchies in  vertex and edge  form  were introduced  \cite{Nijhoff:1992,Ni1}, furthermore, modifications and extensions of their lower order members as well as the hierarchies themselves can be found in \cite{Tongas:2004,Maruno:2010,Hietarinta:2011,JamesPhd,Nijhoff:2011,Atkinson:2012,Zhang:2012,Fordy_2017}.
 On top of that, it is shown here that both $\mathcal{K}^{(i)}$  hierarchies could also be obtained through periodic reductions of   deformed versions of the non-Abelian Hirota-Miwa system. Note that  the non-Abelian  Hirota-Miwa system was firstly introduced in \cite{Nimmo:2006},  in the Abelian setting it was earlier introduced  in \cite{Noumi:2002}.

Each one of the $\mathcal{K}^{(i)}$  hierarchies gives rise to two hierarchies of integrable vertex systems i.e. integrable difference systems with variables defined of the vertices of an elementary cell of the $\mathbb{Z}^2$ graph. It turns out that the vertex systems associated with $\mathcal{K}^{(0)}$ are point equivalent to the ones associated with the $\mathcal{K}^{(1)}$ hierarchy. So from $\mathcal{K}^{(0)}$ we obtain  the non-Abelian lattice-modified Gel'fand-Dikii hierarchy as well as the explicit form of  a hierarchy that we refer to as the {\em non-Abelian lattice-$NQC$ (or lattice-$(Q3)_0$) Gel'fand-Dikii} hierarchy. The non-Abelian lattice-modified Gel'fand-Dikii hierarchy was introduced in \cite{Doliwa_2013} whereas in the Abelian case it was firstly implicitly provided in \cite{Nijhoff:1992} and explicitly in \cite{Atkinson:2012}. The first  member of the non-Abelian lattice-$NQC$ Gel'fand-Dikii hierarchy is the so-called $NQC$ integrable lattice equation that was firstly introduced in \cite{nij-qui-cap}, cf. also \cite{nij-qui-cap2}. Note that the $NQC$ integrable lattice equation is gauge equivalent to the lattice equation that is referred as $(Q3)_0$  in \cite{ABS}.
The whole hierarchy  in the Abelian case,  was implicitly provided in  \cite{Nijhoff:2011}, whereas  its second member i.e. the Boussinesq analogue of $(Q3)_0,$  was explicitly derived in \cite{Zhang:2012}.

We start this manuscript with a brief  introduction. We continue in  Section \ref{Section:2}, where we present the basic notions and definitions used throughout this paper. In addition, we recall from the literature two Lax matrices that play a crucial role to this work. In Section \ref{Section:3}, we deform these Lax matrices and  we obtain two non-Abelian integrable hierarchies of $3D-$compatible maps, the $\mathcal{K}^{(i)},$ $i=0,1$ hierarchies.  Furthermore, we prove their multidimensional compatibility  and we provide implicitly their corresponding Yang-Baxter maps. These hierarchies of maps, are naturally associated with  non-Abelian hierarchies of integrable difference systems with variables defined on the edges of an elementary cell of the $\mathbb{Z}^2$ graph. We proceed with Section \ref{Section:4}, where we associate with the $\mathcal{K}^{(0)}$ hierarchy, two integrable hierarchies of non-Abelian difference systems defined on  the vertices of an elementary cell of the $\mathbb{Z}^2$ graph. Thus we obtain explicitly  the non-Abelian lattice-modified  and the lattice-$(Q3)_0$ Gel'fand-Dikii hierarchies. Finally in Section \ref{Section:6}, we present some ideas on further  research. We conclude this article with  \ref{app1} where we present  non-Abelian  forms of the lattice-potential KdV equation.

\section{Notation, definitions and the  Lax matrices $L^{(N,j)}$ and $M^{(N,j)},$ $j=1,\ldots,N-1,$  $N\geq 2 \in \mathbb{N}$} \label{Section:2}

Here we present the basic objects and definitions that will be considered in this paper. Firstly, let $\mathcal{S}$ be any set. We proceed with the following definitions.

\begin{definition} \label{def:1:1}
The  maps $R: \mathcal{S} \times \mathcal{S} \rightarrow \mathcal{S} \times \mathcal{S}$ and $\widehat R: \mathcal{S} \times \mathcal{S} \rightarrow \mathcal{S} \times \mathcal{S}$  will be called  equivalent if there exists a bijection $\kappa: \mathcal{S} \rightarrow \mathcal{S}$ such that $(\kappa\times \kappa) \circ R=\widehat R \circ (\kappa\times \kappa).$
\end{definition}

 \begin{definition}[$3D-$compatible/consistent maps \cite{ABS:YB}]
Let $Q: \mathcal{S} \times \mathcal{S}\ni({\bf x},{\bf y})\mapsto ({\bf u}, {\bf v})=(f({\bf x},{\bf y}),g({\bf x},{\bf y})) \in \mathcal{S} \times \mathcal{S},$ be a map and  $Q_{ij}$  $i\neq j\in\{1,2,3\},$ be the maps that act as $Q$ on the $i-$th and $j-$th factor of $\mathcal{S} \times \mathcal{S}\times \mathcal{S}$ and as identity to the remaining factor.  In detail we have
\begin{align*}
Q_{12}:({\bf x},{\bf y},{\bf z})\mapsto(\widehat {\bf x},\widetilde {\bf y},{\bf z})=(f({\bf x},{\bf y}),g({\bf x},{\bf y}),{\bf z}),\\
Q_{13}:({\bf x},{\bf y},{\bf z})\mapsto(\bar {\bf x},{\bf y},\widetilde {\bf z})=(f({\bf x},{\bf z}),{\bf y},g({\bf x},{\bf z})),\\
Q_{23}:({\bf x},{\bf y},{\bf z})\mapsto({\bf x},\bar {\bf y},\widehat {\bf z})=({\bf x},f({\bf y},{\bf z}),g({\bf y},{\bf z})).
\end{align*}
The map $Q: \mathcal{S} \times \mathcal{S}\rightarrow \mathcal{S} \times \mathcal{S}$ will be called {\em 3D-compatible} or {\em 3D-consistent map}  if it holds
${\bar {\widehat {\bf x}}}={\widehat {\bar {\bf x}}},$  ${\bar {\widetilde {\bf y}}}={\widetilde {\bar {\bf y}}},$ ${\widehat {\widetilde {\bf z}}}={\widetilde {\widehat {\bf z}}},$ that is
\begin{align}\label{3d:comp:def1}
f(\bar {\bf x},\bar {\bf y})=f(\widehat {\bf x},\widehat {\bf z}),&&\mbox{or}&&f\left(f({\bf x},{\bf z}),f({\bf y},{\bf z})\right)=f\left(f({\bf x},{\bf y}),g({\bf y},{\bf z})\right),
\end{align}
\begin{align} \label{3d:comp:def2}
g(\bar {\bf x},\bar {\bf y})=f(\widetilde {\bf y},\widetilde {\bf z}),&&\mbox{or}&&g\left(f({\bf x},{\bf z}),f({\bf y},{\bf z})\right)=f\left(g({\bf x},{\bf y}),g({\bf x},{\bf z})\right),
\end{align}
\begin{align} \label{3d:comp:def3}
g(\widehat {\bf x},\widehat {\bf z})=g(\widetilde {\bf y},\widetilde {\bf z}),&&\mbox{or}&&g\left(f({\bf x},{\bf y}),g({\bf y},{\bf z})\right)=g\left(g({\bf x},{\bf y}),g({\bf x},{\bf z})\right).
\end{align}
\end{definition}

Relations (\ref{3d:comp:def1})-(\ref{3d:comp:def3}) serve as the compatibility (consistency) relations of the map $Q$ on the cube (see Figure \ref{1st_figure_l}). Consequently, they imply compatibility on any $3$ dimensional  face of the $n$ dimensional cube $\mathbb{Q}_n.$ We refer to this property as multidimensional compatibility or equivalently as multidimensional consistency property.

\begin{definition}[Quadrirational maps and their companion maps \cite{et-2003,ABS:YB}]
A map $R: \mathcal{S} \times \mathcal{S}\ni({\bf x},{\bf y})\mapsto ({\bf u},{\bf v}) \in \mathcal{S} \times \mathcal{S}$   will be called {\em quadrirational}, if both the map $R$ and the so-called {\em companion map} $cR: \mathcal{S} \times \mathcal{S}\ni({\bf x},{\bf v})\mapsto ({\bf u},{\bf y}) \in \mathcal{S} \times \mathcal{S},$ are birational maps.
\end{definition}

An alternative notion that incorporates the $3D-$compatibility property of a map,  is the so-called  {\em Yang-Baxter property}. The maps that satisfy the Yang-Baxter property will be called {\em Yang-Baxter maps}. Note that if a $3D-$compatible map is quadrirational, its   companion map is a  Yang-Baxter map.


\begin{definition}[Yang-Baxter maps \cite{Sklyanin:1988,Drinfeld:1992}]
A map $R: \mathcal{S} \times \mathcal{S}\ni({\bf x},{\bf y})\mapsto ({\bf u}, {\bf v})=(s({\bf x},{\bf y}),t({\bf x},{\bf y}))\in \mathcal{S}\times \mathcal{S},$ will be called a {\em Yang-Baxter map} if it satisfies 
\begin{align} \label{YANG_BAXTER}
R_{12}\circ R_{13}\circ R_{23}= R_{23}\circ R_{13}\circ R_{12},
\end{align}
where $R_{ij}$ $i\neq j\in\{1,2,3\},$ denotes the maps that act as  $R$ on the $i-$th and the $j-$th factor of $\mathcal{S}\times \mathcal{S}\times \mathcal{S},$ and as identity to the remaining factor.
\end{definition}

 Yang-Baxter maps serve as {\em set-theoretical-solutions} of the functional Yang-Baxter equation (\ref{YANG_BAXTER})  and the first instances of such solutions  appeared in \cite{Sklyanin:1988,Drinfeld:1992}.  Note that the term {\em Yang-Baxter maps} was introduced in \cite{Bukhshtaber:1998,Veselov:20031}.

  Yang-Baxter property, as a compatibility property, uses another set of initial data on the cube than the $3D-$ compatibility property (see Figure \ref{1st_figure_l}). In that respect, when the maps $R,Q$ are quadrirational and $R=cQ,$ the Yang-Baxter property is equivalent to the  $3D-$compatibility property. So, as soon as a Yang-Baxter map is quadrirational, its companion is $3D-$compatible and vice versa, irrespectively of the  underlying sets on which the map acts. Note that in general the companion map $cQ$ of  a $3D-$compatible map $Q$  has different functional form than $Q,$ so it  cannot be a Yang-Baxter map at the same time.  There exist though cases  where the companion map $cQ$ of a   $3D$-compatible map $Q$ coincides (has the same functional form) with  $Q,$  then the map $Q$ shares both the Yang-Baxter and the $3D$-compatibility property \cite{ABS:YB}.

 \setlength{\perspective}{14pt}
    \def\isofactor{0.5}
    \begin{figure}[h]
\begin{minipage}[htb]{0.45\textwidth}
\adjustbox{scale=0.99,center}{
\begin{tikzcd}[every arrow/.append style={}, row sep={38,between origins}, column sep={38,between origins}]
\arrow[dashrightarrow, from=7-1,to=1-2,"Q_{23}" {description},crossing over] \arrow[dashrightarrow,from=5-2,to=2-4,"Q_{13}" {description},crossing over, shorten=6mm, end anchor={[yshift=-5mm]north east}] \arrow[dashrightarrow,from=7-1,to=6-4,"Q_{12}" {description},crossing over] 
    &[-\perspective] |[alias=A]| \ar[dash,"{\bar {\widehat x}}" {description},shorten <= -5,shorten >= -5]{rrd}\ar[dash,"{\widehat z}" {description},shorten <= -6.5,shorten >= -1]{dddd}\ar[dash,"{\bar y}" {description},shorten <= -5,shorten >= -2]{ddl}  &[\perspective] &[-\perspective] \\[-38+\isofactor\perspective]
    & & & {} \ar[dash,"{\widehat {\widetilde z}}" {description},shorten <= -6.5,shorten >= -1]{dddd}\ar[dash,"{\bar {\widetilde y}}" {description},shorten <= -7,shorten >= -8]{ddl}  &   \\[-\perspective-\isofactor\perspective]
    |[alias=Z]|  \ar[dash,crossing over,"{\bar x}" {description},shorten <= -5,shorten >= -5]{rrd} \ar[dash,"z" {description},shorten <= -6.5,shorten >= -1]{dddd} &  &  &   \\[-38+\isofactor\perspective]
    & &  {} \\[\perspective-\isofactor\perspective]
    & {}  \ar[dash,"{\widehat x}" {description},shorten <= -5,shorten >= -5]{rrd} \ar[dash,"y" {description},shorten <= -6,shorten >= -2]{ddl} & &  \\[-38+\isofactor\perspective]
    & & & {} \ar[dash,"{\widetilde y}" {description},shorten <= -6,shorten >= -2]{ddl}  \\[-\perspective-\isofactor\perspective]
     \ar[dash,"x" {description},shorten <= -5,shorten >= -5]{rrd}   & &  \\[-38+\isofactor\perspective]
    {} &  & {} \ar[from=uuuu,crossing over,dash,"{\widetilde z}" {description},shorten <= -6.5,shorten >= -1]  &
    \end{tikzcd}}
    \captionsetup{font=footnotesize}
\captionof*{figure}{(a) }
\end{minipage}
\begin{minipage}[htb]{0.45\textwidth}
\adjustbox{scale=0.99,center}{
\begin{tikzcd}[every arrow/.append style={}, row sep={38,between origins}, column sep={38,between origins}]
    &[-\perspective] |[alias=A]| \ar[dash,"{\widehat {\bar x}}" {description},shorten <= -5,shorten >= -5]{rrd}\ar[dash,"{\widetilde z}" {description},shorten <= -6.5,shorten >= -1]{dddd}\ar[dash,"{\widetilde {\bar y}}" {description},shorten <= -5,shorten >= -2]{ddl}  &[\perspective] &[-\perspective] \\[-38+\isofactor\perspective]
    & & & {} \ar[dash,"z" {description},shorten <= -6.5,shorten >= -1]{dddd}\ar[dash,"{\bar y}" {description},shorten <= -7,shorten >= -8]{ddl}  &   \\[-\perspective-\isofactor\perspective]
    |[alias=Z]|  \ar[dash,crossing over,"{\bar x}" {description},shorten <= -5,shorten >= -5]{rrd} \ar[dash,"{\widetilde{\widehat z}}" {description},shorten <= -6.5,shorten >= -1]{dddd} &  & \arrow[to=A,"R_{12}" {description},crossing over] &   \\[-38+\isofactor\perspective]
    & &  {} \\[\perspective-\isofactor\perspective]
    & {}  \ar[dash,"{\widehat x}" {description},shorten <= -5,shorten >= -5]{rrd} \ar[dash,"{\widetilde y}" {description},shorten <= -6,shorten >= -2]{ddl} & &  \\[-38+\isofactor\perspective]
    & & & {} \ar[dash,"y" {description},shorten <= -6,shorten >= -2]{ddl} \arrow["R_{23}" {description}]{uul} \\[-\perspective-\isofactor\perspective]
     \ar[dash,"x" {description},shorten <= -5,shorten >= -5]{rrd}   & &  \\[-38+\isofactor\perspective]
    {} &  & {} \ar[from=uuuu,crossing over,dash,"{\widehat z}" {description},shorten <= -6.5,shorten >= -1] \arrow[to=Z,"R_{13}" {description},crossing over]  &
    \end{tikzcd}}
    \captionsetup{font=footnotesize}
\captionof*{figure}{(b) }
\end{minipage}
\caption{(a): Left hand side of the $3D-$compatibility formulas (\ref{3d:comp:def1}), (\ref{3d:comp:def3}), that is \newline
$({\bf x},{\bf y},{\bf z}) \xmapsto[]{Q_{12}} ({\widehat {\bf x}},{\widetilde {\bf y}},{\bf z}),$ $({\bf x},{\bf y},{\bf z}) \xmapsto[]{Q_{23}} ( {\bf x},{\bar {\bf y}},{\widehat {\bf z}}),$ $({\widehat {\bf x}},{\bf y},{\widehat {\bf z}}) \xmapsto[]{Q_{13}} ({\bar {\widehat  {\bf x}}}, {\bf y},{\widetilde {\widehat {\bf z}}}).$ \\ (b): Right hand side of the Yang-Baxter equation, that is \newline
$
({\bf x},{\bf y},{\bf z}) \xmapsto[]{R_{23}} ({\bf x},{\bar {\bf y}},{\widehat {\bf z}}) \xmapsto[]{R_{13}} ({\bar {\bf x}},{\bar {\bf y}},{\widetilde {\widehat {\bf z}}}) \xmapsto[]{R_{12}}  ({\widehat {\bar {\bf x}}},{\widetilde {\bar {\bf y}}},{\widetilde {\widehat {\bf z}}}).
$
}\label{1st_figure_l}
\end{figure}


\begin{definition}
A bijection $\phi: \mathcal{S}\rightarrow \mathcal{S}$ will be called a {\em symmetry} of the    map $R: \mathcal{S} \times \mathcal{S}\rightarrow \mathcal{S} \times \mathcal{S},$ if  $(\phi\times \phi)  \circ R= R  \circ (\phi\times \phi).$
\end{definition}

\begin{definition}[\cite{Veselov:2003b,Nijhoff:2002}]
The matrix $L({\bf x};\lambda)$
 will be called  the  {\em Lax matrix} of the  $3D-$compatible map $H: ({\bf x},{\bf y})\mapsto ({\bf u},{\bf v}),$  if the relation
\begin{align} \label{laxccc}
L({\bf u};\lambda)L({\bf y};\lambda)=L({\bf v};\lambda)L({\bf x};\lambda)
\end{align}
implies  for all $\lambda$ that
$
H({\bf x},{\bf y})=({\bf u},{\bf v}).
$
$L({\bf x};\lambda)$ will be called a {\em strong Lax matrix}, if mapping $H$ is implied uniquely from (\ref{laxccc}).
\end{definition}

 Next, we define the order $N$ lower and upper-triangular nilpotent matrices.
\begin{definition}\label{np_mat}
With $\nabla^k,$ and $\Delta^k,$ $k=1,2,\ldots, N-1,$ we respectively define the order $N$ lower-triangular and upper-triangular nilpotent matrices i.e.
\begin{align*}
\left(\nabla^{(k)}\right)_{ij}:=\left\{ \begin{array}{ll}
            0,& i\leq j \\
            \delta_{i,j+k},& i>j
                \end{array} \right.&& \left(\Delta^{(k)}\right)_{ij}:=\left\{ \begin{array}{ll}
            \delta_{i+N-k,j},& i< j \\
              0,& i\geq j
                \end{array} \right.
\end{align*}
\end{definition}
 Finally, we define the notion of {\em Deformation matrices}.
\begin{definition}[Deformation matrix]\label{def_def_mat}
Let  $L({\bf x};\lambda)$ be an order $N$ strong
 Lax matrix of the  $3D-$compatible map $H: ({\bf x},{\bf y})\mapsto ({\bf u},{\bf v}).$ Fix $k\in\mathbb{N}$ with $0\leq k< N.$ The   matrix $D^{(k)}({\bf x},{\boldsymbol \alpha}),$ where ${\boldsymbol \alpha}$ a collection of constants, will be called {\em deformation matrix} and the constants ${\boldsymbol \alpha}$ deformation constants, if it holds:
 \begin{enumerate}
 \item $D^{(k)}({\bf x},{\bf 0})=\Delta^{(k)}+\nabla^{(k)},$ where $\Delta^{(k)},$ and $\nabla^{(k)}$ the order $N$ nilpotent matrices of Definition \ref{np_mat}, and ${\bf 0}$ a collection of zeros,
 \item  ${\widehat L}({\bf x};\lambda):= D^{(k)}({\bf x},{\boldsymbol \alpha}) L({\bf x};\lambda)$ serves as a strong Lax matrix for a  family of  maps ${\widehat H({\boldsymbol \alpha})}: ({\bf x},{\bf y})\mapsto ({\widehat {\bf u}},{\widehat {\bf v}}).$
 \end{enumerate}
 The Lax matrix ${\widehat L}({\bf x};\lambda)$ will be called {\em deformed Lax matrix}.
 \end{definition}
 \begin{remark} \label{First_Remark}
 When $k=0,$  for the deformation matrix $D^{(0}({\bf x},{\boldsymbol \alpha}),$ it holds
that $D^{(0)}({\bf x},{\bf 0})=I,$ where $I$ the order $N$ identity matrix, and clearly there follows ${\widehat H({\bf 0})}\equiv H.$ The deformation matrix $D^{(0)}({\bf x},{\boldsymbol \alpha})$ will be called {\em diagonal deformation matrix}. Note that a diagonal deformation matrix is not necessarily a diagonal matrix.
 \end{remark}

There is a natural correspondence of a map with a difference system  defined on the edges of an elementary quad of the $\mathbb{Z}^2$ graph. Specifically, a map $R: \mathcal{S} \times \mathcal{S}\ni (\bf x,\bf y) \mapsto  (\bf u,\bf v) \in  \mathcal{S} \times \mathcal{S},$ can be considered as a difference system defined on the edges of an elementary quadrilateral of the $\mathbb{Z}^2$ graph  by  making the following identifications
\begin{align} \label{notation1}
\begin{aligned}
{\bf x}\equiv  {\bf x}_{m+1/2,n},&&{\bf y}\equiv  {\bf y}_{m,n+1/2}\\
{\bf u}\equiv  {\bf x}_{m+1/2,n+1},&&{\bf v}\equiv  {\bf y}_{m+1,n+1/2},
\end{aligned}, && m,n\in \mathbb{Z}.
\end{align}
Moreover, we can adopt the compendious notation (see Figure \ref{fig1})
\begin{align*}
\begin{aligned}
{\bf x}:={\bf x}_{m+1/2,n},&& {\bf y}:={\bf y}_{m,n+1/2}, && {\bf x}_1:={\bf x}_{m+3/2,n}, && etc.\\
 {\bf x}_2:={\bf x}_{m+1/2,n+1}\equiv {\bf u},  &&{\bf y}_1:={\bf y}_{m+1,n+1/2}\equiv {\bf v}, && {\bf y}_2:={\bf y}_{m,n+3/2},&& etc.
\end{aligned}, && m,n\in \mathbb{Z}.
\end{align*}
Note that ${\bf x}$ could be a collection of variables, so ${\bf x}=\left(x^{(1)},\ldots, x^{(N)}\right),$ $N\in \mathbb{N}.$ In this case with ${\bf x}_2$ we denote ${\bf x}_2=\left(x^{(1)}_2,\ldots, x_2^{(N)}\right),$ that is all variables from the collection shifted in the second direction and similarly for ${\bf y}_1$.  So, unless otherwise stated and for the rest of this article, we represent the components of a vector with superscripts inside parentheses. When the superscripts denote powers,  we do not use parentheses.

\begin{figure}[h]
\begin{center}
\begin{minipage}[htb]{0.4\textwidth} \tikzcdset{every label/.append style = {font = \large}}
\begin{tikzcd}[row sep=1.5in, column sep = 1.5in,every arrow/.append style={dash}]
  \arrow[d,"\bff{y}\equiv\bff{y}_{m,n+1/2}",sloped,anchor=north]  \bff{\phi}_{m,n+1}  \arrow{r}{ \bff{u}\equiv\bff{x}_{m+1/2,n+1}}&   \bff{\phi}_{m+1,n+1}\arrow{d} \\
   \bff{\phi}_{m,n} \arrow{r}{ \bff{x}\equiv\bff{x}_{m+1/2,n}}&  \bff{\phi}_{m+1,n} \arrow[u," \bff{v}\equiv\bff{y}_{m+1,n+1/2}",sloped]
\end{tikzcd}
\captionsetup{font=footnotesize}
\captionof*{figure}{(a) Descriptive notation}
\end{minipage}\hspace{0.9cm}
\begin{minipage}[htb]{0.4\textwidth} \tikzcdset{every label/.append style = {font = \large}}
\begin{tikzcd}[row sep=1.5in, column sep = 1.5in,every arrow/.append style={dash}]
  {\bff \phi}_2 \arrow{r}{{\bf x}_2}\arrow{d}{{\bf y}}&  {\bff \phi}_{12}\arrow{d} \\
  {\bff \phi} \arrow{r}{{\bf x}}& {\bff \phi}_1\arrow{u}{{\bf y}_1}
\end{tikzcd}
\captionsetup{font=footnotesize}
\captionof*{figure}{(b) Compendious notation}
\end{minipage}\
\caption{Variables assigned on vertices and edges of  an elementary cell of the $\mathbb{Z}^2$ graph} \label{fig1}
\end{center}
\end{figure}
Let $\mathbb{A}$ be an  associative algebra over a field $\mathbb{F},$  with multiplicative identity that we denote with $1.$
Throughout this paper we consider $\mathcal{S}=\underbrace{\mathbb{A}^\times\times\cdots \times \mathbb{A}^\times}_{N-\text{times}},$ $N\in \mathbb{N},$ where $\mathbb{A}^\times$ denotes the subgroup of elements $w\in \mathbb{A}$ having multiplicative inverse $w^{-1}\in \mathbb{A},$ s.t. $w w^{-1}=w^{-1} w=1.$ In addition, with $C(\mathbb{A}^\times)$ we denote the center of algebra $\mathbb{A}^\times$ i.e. a commutative subalgebra of $\mathbb{A}^\times$ consisting of invertible elements.

In this general setting, $\mathbb{A}^\times$ could be a division ring for instance bi-quaternions.  More generally, $\mathbb{A}^\times$ could stand for the subgroup of invertible matrices of the algebra $\mathbb{A}$ of $n\times n$  matrices.

\subsection{The  Lax matrices $ L^{(N,j)},$ $M^{(N,j)},$ $j=1,\ldots,N-1$  $N\geq 2 \in \mathbb{N}$ and integrable hierarchies of difference systems}

The following Lax matrices of order N, have appeared in various occasions and in different context inside the theory of integrable systems, see for instance  \cite{Fordy:1980,Bog1,Nijhoff:1996,Kassotakis3:2020}. One of these Lax matrices, in particular the  Lax matrix that in what follows we denote as $ L^{(N,1)},$ was firstly considered in \cite{Nijhoff:1992} in connection with the lattice Gel'fand-Dikii hierarchy in the Abelian setting. After this seminal work, modifications of this Lax matrix led to various hierarchies of integrable difference systems in the Abelian \cite{Atkinson:2012,Fordy_2017}  and recently in the non-Abelian domain \cite{Doliwa_2013,Kassotakis:1:2021}.  The Lax matrix $ L^{(N,1)}$ explicitly reads

\begin{align}\label{Lax:eq11}
 L^{(N,1)}({\bf x};\lambda):= I_N+\nabla^{(1)} {\bf X}+\lambda \Delta^{(1)} {\bf X}=
\begin{pmatrix}
                      1&0&\cdots&0&\lambda\,x^{(N)}\\
                      x^{(1)}&1& 0& \cdots &0 \\
                      0  &x^{(2)}&\ddots& {}  &\vdots\\
                      \vdots & &\ddots &1 &0\\
                       0     & 0 &   &  x^{(N-1)}    &1
                     \end{pmatrix},
\end{align}
 and serves as a specific member (j=1) of the family of Lax matrices $ L^{(N,j)}({\bf x};\lambda):= I_N+\nabla^{(j)} {\bf X}+\lambda \Delta^{(j)} {\bf X},$ $j=1,\ldots, N-1,$
where $I_N$ the order $N$ identity matrix and   ${\bf X}$ an order $N$ diagonal matrix with entries $({\bf X})_{i,i}=x^{(i)}.$

  The compatibility conditions  $L^{(N,1)}({\bf u};\lambda)L^{(N,1)}({\bf y};\lambda)=L^{(N,1)}({\bf v};\lambda)L^{(N,1)}({\bf x};\lambda),$ read
  \begin{align*}
  u^{(i)}y^{(i-1)}=v^{(i)}x^{(i-1)},&&u^{(i)}+y^{(i)}=v^{(i)}+x^{(i)},&&i=1,2,\ldots, N,
  \end{align*}
  and they assure that the
   the Lax matrix $L^{(N,1)}$ serves as a strong Lax matrix for the hierarchy of maps
\begin{align} \label{map:G}
\mathcal{G}: \left(x^{(1)},\ldots, x^{(N)},y^{(1)},\ldots, y^{(N)}\right)\mapsto  \left(u^{(1)},\ldots, u^{(N)},v^{(1)},\ldots, v^{(N)}\right),
\end{align}
\begin{align*}
\begin{aligned}
u^{(i)}=\left(x^{(i)}-y^{(i)}\right)x^{(i-1)}\left(x^{(i-1)}-y^{(i-1)}\right)^{-1},\\
v^{(i)}=\left(y^{(i)}-x^{(i)}\right)y^{(i-1)}\left(y^{(i-1)}-x^{(i-1)}\right)^{-1},
\end{aligned}& & i=1,2,\ldots, N.
\end{align*}
This hierarchy in the non-Abelian setting was firstly considered in \cite{Kassotakis:1:2021}, and in vertex variables depending on a two-fold choice of potential functions, serves as  the lattice-modified or the lattice-Schwarzian Gel'fand-Dikii hierarchy. Moreover its companion hierarchy of maps, defines the hierarchy of $H^A_{III}$ Yang-Baxter maps.
It is easy to show that $\mathcal{G},$ has as symmetry the bijection
\begin{align*}
\psi: \left(x^{(1)},x^{(2)},\ldots, x^{(N)}\right)\mapsto  \left(x^{(N)},x^{(1)},\ldots, x^{(N-1)}\right).
\end{align*}

The family of Lax matrices $ M^{(N,j)}({\bf x};\lambda):= {\bf X}+\nabla^{(j)} +\lambda \Delta^{(j)},$ $j=1,\ldots,N-1,$ can be considered as a dual family to the family of Lax matrices $L^{(N,j)}.$   For the detailed study, in the Abelian case, of such family of Lax matrices we refer to \cite{Fordy_2017}. 
Here we consider one member of this family, namely, $ M^{(N,1)}$ that  explicitly reads:
\begin{align}\label{Lax:eq00}
 M^{(N,1)}({\bf x};\lambda):= {\bf X}+\nabla^{(1)} +\lambda \Delta^{(1)}=
\begin{pmatrix}
                      x^{(N)}&0&\cdots&0&\lambda\\
                      1&x^{(1)}& 0& \cdots &0 \\
                      0  &1&\ddots& {}  &\vdots\\
                      \vdots & &\ddots &x^{(N-2)} &0\\
                       0     & 0 &   &  1    &x^{(N-1)}
                     \end{pmatrix}.
\end{align}
 The compatibility conditions  $M^{(N,1)}({\bf u};\lambda)M^{(N,1)}({\bf y};\lambda)=M^{(N,1)}({\bf v};\lambda)M^{(N,1)}({\bf x};\lambda),$ read
  \begin{align*}
  u^{(i)}y^{(i)}=v^{(i)}x^{(i)},&&u^{(i)}+y^{(i-1)}=v^{(i)}+x^{(i-1)},&&i=1,2,\ldots, N,
  \end{align*}
  and they assure that    the Lax matrix $M^{(N,1)}$ serves as a strong Lax matrix for the hierarchy of maps
%
%
%
%
\begin{align} \label{map:D}
\mathcal{D}: \left(x^{(1)},\ldots, x^{(N)},y^{(1)},\ldots, y^{(N)}\right)\mapsto  \left(u^{(1)},\ldots, u^{(N)},v^{(1)},\ldots, v^{(N)}\right),
\end{align}
\begin{align*}
\begin{aligned}
u^{(i)}=\left(x^{(i-1)}-y^{(i-1)}\right)x^{(i)}\left(x^{(i)}-y^{(i)}\right)^{-1},\\ v^{(i)}=\left(y^{(i-1)}-x^{(i-1)}\right)y^{(i)}\left(y^{(i)}-x^{(i)}\right)^{-1},
\end{aligned}& & i=1,2,\ldots, N.
\end{align*}
The non-Abelian hierarchy of maps $\mathcal{D}$ was firstly considered in \cite{Doliwa_2013,Doliwa_2014},  where     it was also considered as a periodic reduction of the non-Abelian Hirota-Miwa system \cite{Nimmo:2006}. Furthermore, in \cite{Doliwa_2014} it was shown that the hierarchy of companion  maps of $\mathcal{D}$, defines the hierarchy of $H^B_{III}$ Yang-Baxter maps.  In vertex variables, $\mathcal{D}$ defines either the lattice-modified or the lattice-Schwarzian Gel'fand-Dikii hierarchy.

\section{Deformed  Lax matrices  and integrable hierarchies of difference systems} \label{Section:3}

Following Definition \ref{def_def_mat} and in particular Remark \ref{First_Remark}, here  we are searching for families of diagonal deformation matrices  for the Lax matrices $ L^{(N,1)}$ and  $M^{(N,1)},$ of the previous Section. In detail,  we search for diagonal deformation matrices $D^{(0)}({\bf x},{\boldsymbol \alpha},{\boldsymbol \beta}),$ under the additional requirements, first, they are diagonal matrices of order $N$, second, their entries read
\begin{align} \label{search_d}
\left(D^{(0)}({\bf x},{\boldsymbol \alpha},{\boldsymbol \beta})\right)_{i,i}=\left(\alpha^{(i-1)}-\beta^{(i-1)} x^{(i-1)}\right)^{-1},&&i=1,\ldots,N, && \alpha^{(i)},\beta^{(i)} \in C(\mathbb{A}^\times).
\end{align}
In order for  $D^{(0)}({\bf x},{\boldsymbol \alpha},{\boldsymbol \beta})$ to be a deformation matrix for the instance of  the strong Lax matrix $M^{(N,1)}$, a necessary condition (see Definition \ref{def_def_mat}) is that the deformed Lax  matrix ${\widehat M}^{(N,1)}({\bf x};\lambda):=D^{(0)}({\bf x},{\boldsymbol \alpha},{\boldsymbol \beta}) M^{(N,1)}({\bf x};\lambda)$  remains a strong Lax matrix, that is it defines uniquely a  family of maps. Clearly that is a strong requirement that will put some conditions on  the the deformation parameters. Indeed, the Lax equation $\widehat M^{(N,1)}({\bf u};\lambda)\widehat M^{(N,1)}({\bf y};\lambda)=\widehat M^{(N,1)}({\bf v};\lambda)\widehat M^{(N,1)}({\bf x};\lambda),$ provides the following three sets of equations, where each set consists of  $N$ equations,
\begin{align*}
\begin{aligned}
\left(\alpha^{(i)}-\beta^{(i)} u^{(i)}\right)^{-1} \left(\alpha^{(i-1)}-\beta^{(i-1)} y^{(i-1)}\right)^{-1}=\left(\alpha^{(i)}-\beta^{(i)} v^{(i)}\right)^{-1} \left(\alpha^{(i-1)}-\beta^{(i-1)}x^{(i-1)}\right)^{-1},
\end{aligned}\\
\begin{aligned}
\left(\alpha^{(i)}-\beta^{(i)} u^{(i)}\right)^{-1}u^{(i)} \left(\alpha^{(i)}-\beta^{(i)} y^{(i)}\right)^{-1}y^{(i)}=\left(\alpha^{(i)}-\beta^{(i)} v^{(i)}\right)^{-1}v^{(i)} \left(\alpha^{(i)}-\beta^{(i)} x^{(i)}\right)^{-1}x^{(i)},
\end{aligned}\\
\begin{aligned}
&\left(\alpha^{(i)}-\beta^{(i)} u^{(i)}\right)^{-1} \left(\alpha^{(i-1)}-\beta^{(i-1)} y^{(i-1)}\right)^{-1}y^{(i-1)}+\left(\alpha^{(i)}-\beta^{(i)} u^{(i)}\right)^{-1}u^{(i)} \left(\alpha^{(i)}-\beta^{(i)} y^{(i)}\right)^{-1}\\
&=\left(\alpha^{(i)}-\beta^{(i)} v^{(i)}\right)^{-1} \left(\alpha^{(i-1)}-\beta^{(i-1)} x^{(i-1)}\right)^{-1}x^{(i-1)}+\left(\alpha^{(i)}-\beta^{(i)} v^{(i)}\right)^{-1}v^{(i)} \left(\alpha^{(i)}-\beta^{(i)} x^{(i)}\right)^{-1},
\end{aligned}
\end{align*}
and the superscript $i=1,\ldots, N,$ is considered modulo $N$.  If we demand that the three sets of equations above are linearly dependent,  
we obtain that $\beta^{(1)}=\ldots=\beta^{(N)}=\beta,$  so the diagonal matrices $D^{(0)}({\bf x},{\boldsymbol \alpha},{\boldsymbol \beta}),$ with entries $\left(D^{(0)}({\bf x},{\boldsymbol \alpha},{\boldsymbol \beta})\right)_{i,i}=\left(\alpha^{(i)}-\beta x^{(i-1)}\right)^{-1}$ serves as a family of diagonal deformation matrices for the strong Lax matrix $M^{(N,1)}$ and the deformed Lax equation implies as unique solution the following family of hierarchies of maps
$$
\left(x^{(1)},\ldots, x^{(N)},y^{(1)},\ldots, y^{(N)}\right)\mapsto  \left(u^{(1)},\ldots, u^{(N)},v^{(1)},\ldots, v^{(N)}\right),
$$
where:
\begin{align} \label{SEC3:par_011}
\begin{aligned}
u^{(i)}=\left(\alpha^{(i-1)}-\beta y^{(i-1)}\right)^{-1}\left(x^{(i-1)}-y^{(i-1)}\right)x^{(i)}\left(x^{(i)}-y^{(i)}\right)^{-1}\left(\alpha^{(i)}-\beta y^{(i)}\right),\\
v^{(i)}=\left(\alpha^{(i-1)}-\beta x^{(i-1)}\right)^{-1}\left(y^{(i-1)}-x^{(i-1)}\right)y^{(i)}\left(y^{(i)}-x^{(i)}\right)^{-1}\left(\alpha^{(i)}-\beta x^{(i)}\right),\\
i=1,2,\ldots, N.
\end{aligned}& &
\end{align}
 Note that w.l.o.g. we could set $\alpha^{(i)},\beta\in\{0,1\},$ $i=1,\ldots, N.$ The generic case corresponds to $\alpha^{(1)}=\ldots=\alpha^{(N)}=\beta=1,$ while degenerate cases arise when we equalize to zero some of the  $\alpha^{(i)}.$ Note that in our setting,  the most degenerate case corresponds to set $\beta=0,$ and then the family of hierarchies above, coincides with (\ref{map:D}).

Working similarly, we find that the diagonal matrices $\widetilde D^{(0)}({\bf x},{\boldsymbol \alpha},{\boldsymbol \beta})$ with entries
 \begin{align*}
\left(\widetilde D^{(0)}({\bf x},{\boldsymbol \alpha},{\boldsymbol \beta})\right)_{i,i}=\left(\alpha-\beta^{(i-1)} x^{(i-1)}\right)^{-1},
\end{align*}
serves as a family of diagonal deformation matrices for the strong Lax matrix $L^{(N,1)}.$ We denote the deformed Lax matrix as ${\widehat L}^{(N,1)}({\bf x};\lambda):=\widetilde D^{(0)}({\bf x},{\boldsymbol \alpha},{\boldsymbol \beta}) L^{(N,1)}({\bf x};\lambda)$  and the deformed Lax equation implies as unique solution the following family of hierarchies of maps
$$
\left(x^{(1)},\ldots, x^{(N)},y^{(1)},\ldots, y^{(N)}\right)\mapsto  \left(u^{(1)},\ldots, u^{(N)},v^{(1)},\ldots, v^{(N)}\right),
$$
where:
\begin{align} \label{SEC3:par_012}
\begin{aligned}
u^{(i)}=\left(\alpha-\beta^{(i)} y^{(i)}\right)^{-1}\left(x^{(i)}-y^{(i)}\right)x^{(i-1)}\left(x^{(i-1)}-y^{(i-1)}\right)^{-1}\left(\alpha-\beta^{(i-1)} y^{(i-1)}\right),\\
v^{(i)}=\left(\alpha-\beta^{(i)} x^{(i)}\right)^{-1}\left(y^{(i)}-x^{(i)}\right)y^{(i-1)}\left(y^{(i-1)}-x^{(i-1)}\right)^{-1}\left(\alpha-\beta^{(i-1)} x^{(i-1)}\right),\\
i=1,2,\ldots, N.
\end{aligned}& &
\end{align}
The generic case corresponds to $\beta^{(1)}=\ldots=\beta^{(N)}=\alpha=1,$ while   the most degenerate case corresponds to set $\beta^{(1)}=\ldots=\beta^{(N)}=0,$ and then the family of hierarchies above, coincides with (\ref{map:G}).

Note that for the generic members of both families of hierarchies presented above  the deformation matrices coincide with the following matrix that we denote as $D({\bf x})$
\begin{align}\label{pema}
D({\bf x}):=\begin{pmatrix}
                      \left(1-x^{(N)}\right)^{-1}&0&\cdots&0&0\\
                      0&\left(1-x^{(1)}\right)^{-1}& 0& \cdots &0 \\
                      0  &0&\ddots& {}  &\vdots\\
                      \vdots & &\ddots &\left(1-x^{(N-2)}\right)^{-1} &0\\
                       0     & 0 &   &  0    &\left(1-x^{(N-1)}\right)^{-1}
                     \end{pmatrix}.
\end{align}
A detailed analysis on the generic members of both families of hierarchies presented above, is    presented in the  forthcoming Section.


\subsection{Two hierarchies of $3$D-compatible maps}
In the following  propositions, we  respectively present in detail the generic members of both families of hierarchies of $3D-$compatible maps that correspond to the deformed Lax matrices $\widehat M^{(N,1)}$ and $\widehat L^{(N,1)},$ of the previous Section.


Under identification (\ref{notation1}), $3$D-compatible maps correspond to integrable difference systems with variables defined on the edges of an elementary cell of the $\mathbb{Z}^2$ graph. In that respect, in this Section although we refer to hierarchies of $3$D-compatible maps, at the same time we refer to hierarchies of integrable difference systems with  edge variables.

\begin{prop}\label{prop:01}
The hierarchy of maps
\begin{align*}
\mathcal{K}^{(0)}: \left(x^{(1)},\ldots, x^{(N)},y^{(1)},\ldots, y^{(N)}\right)\mapsto  \left(u^{(1)},\ldots, u^{(N)},v^{(1)},\ldots, v^{(N)}\right),
\end{align*}
where
\begin{align} \label{sec3:01}
\begin{aligned}
u^{(i)}=\left(1-y^{(i-1)}\right)^{-1}\left(x^{(i-1)}-y^{(i-1)}\right)x^{(i)}\left(x^{(i)}-y^{(i)}\right)^{-1}\left(1-y^{(i)}\right),\\
v^{(i)}=\left(1-x^{(i-1)}\right)^{-1}\left(y^{(i-1)}-x^{(i-1)}\right)y^{(i)}\left(y^{(i)}-x^{(i)}\right)^{-1}\left(1-x^{(i)}\right),\\
i=1,2,\ldots, N,
\end{aligned}
\end{align}
\begin{enumerate}
\item has as strong Lax matrix the matrix $\widehat M^{(N,1)}({\bf x};\lambda):=D({\bf x})\, M^{(N,1)}\,({\bf x};\lambda),$ where $M^{(N,1)}({\bf x};\lambda)$  the Lax matrix (\ref{Lax:eq00}) and
 $D({\bf x})$ the deformation matrix (\ref{pema});
\item has as symmetry the bijection
\begin{align*}
\psi: \left(x^{(1)},x^{(2)},\ldots, x^{(N)}\right)\mapsto  \left(x^{(N)},x^{(1)},\ldots, x^{(N-1)}\right).
\end{align*}
\end{enumerate}
\end{prop}

\begin{proof}
Explicitly the Lax matrix $\widehat M^{(N,1)}({\bf x};\lambda)$ reads
\begin{align*}
\widehat M^{(N,1)}({\bf x};\lambda):= 
{\scriptsize \begin{pmatrix}
                      \left(1-x^{(N)}\right)^{-1}x^{(N)}&0&\cdots&0&\lambda\,\left(1-x^{(N)}\right)^{-1}\\
                      \left(1-x^{(1)}\right)^{-1}&\left(1-x^{(1)}\right)^{-1}x^{(1)}& 0& \cdots &0 \\
                      0  &\left(1-x^{(2)}\right)^{-1}&\ddots& {}  &\vdots\\
                      \vdots & &\ddots &\left(1-x^{(N-2)}\right)^{-1}x^{(N-2)} &0\\
                       0     & 0 &   &  \left(1-x^{(N-1)}\right)^{-1}    &\left(1-x^{(N-1)}\right)^{-1}x^{(N-1)}
                     \end{pmatrix}.}
\end{align*}
The compatibility conditions $\widehat M^{(N,1)}({\bf u};\lambda)\widehat M^{(N,1)}({\bf y};\lambda)=\widehat M^{(N,1)}({\bf v};\lambda)\widehat M^{(N,1)}({\bf x};\lambda)$ read:
\begin{align}
\begin{aligned}\label{sec3:02}
\left(1-u^{(i)}\right)^{-1} \left(1-y^{(i-1)}\right)^{-1}=\left(1-v^{(i)}\right)^{-1} \left(1-x^{(i-1)}\right)^{-1},
\end{aligned}\\
\begin{aligned}\label{sec3:03}
\left(1-u^{(i)}\right)^{-1}u^{(i)} \left(1-y^{(i)}\right)^{-1}y^{(i)}=\left(1-v^{(i)}\right)^{-1}v^{(i)} \left(1-x^{(i)}\right)^{-1}x^{(i)},
\end{aligned}\\
\begin{aligned}\label{sec3:04}
&\left(1-u^{(i)}\right)^{-1} \left(1-y^{(i-1)}\right)^{-1}y^{(i-1)}+\left(1-u^{(i)}\right)^{-1}u^{(i)} \left(1-y^{(i)}\right)^{-1}\\
&=\left(1-v^{(i)}\right)^{-1} \left(1-x^{(i-1)}\right)^{-1}x^{(i-1)}+\left(1-v^{(i)}\right)^{-1}v^{(i)} \left(1-x^{(i)}\right)^{-1}
\end{aligned}
\end{align}
where the superscript $i=1,\ldots, N,$ is considered modulo $N$. The three sets of compatibility conditions presented above are not functionally independent since it can be easily shown that equations (\ref{sec3:04}) can be obtained by adding (\ref{sec3:02}) and (\ref{sec3:03}). Then it is easy to verify that the solved form of (\ref{sec3:02}),(\ref{sec3:03}), is exactly (\ref{sec3:01}). In addition, the compatibility conditions (\ref{sec3:02})-(\ref{sec3:04}) are clearly invariant under the map $\phi:=\psi\times \psi,$ that proves that $\psi$ is a symmetry of $\mathcal{K}^{(0)}.$
\end{proof}

\begin{prop}\label{prop:02}
The hierarchy of maps
\begin{align*}
\mathcal{K}^{(1)}: \left(x^{(1)},\ldots, x^{(N)},y^{(1)},\ldots, y^{(N)}\right)\mapsto  \left(u^{(1)},\ldots, u^{(N)},v^{(1)},\ldots, v^{(N)}\right),
\end{align*}
where
\begin{align} \label{sec3:1}
\begin{aligned}
u^{(i)}=\left(1-y^{(i)}\right)^{-1}\left(x^{(i)}-y^{(i)}\right)x^{(i-1)}\left(x^{(i-1)}-y^{(i-1)}\right)^{-1}\left(1-y^{(i-1)}\right),\\
v^{(i)}=\left(1-x^{(i)}\right)^{-1}\left(y^{(i)}-x^{(i)}\right)y^{(i-1)}\left(y^{(i-1)}-x^{(i-1)}\right)^{-1}\left(1-x^{(i-1)}\right),\\
i=1,2,\ldots, N,
\end{aligned}
\end{align}
\begin{enumerate}
\item has as strong Lax matrix the matrix $\widehat L^{(N,1)}({\bf x};\lambda):=D({\bf x})\, L^{(N,1)}\,({\bf x};\lambda),$ where $L^{(N,1)}({\bf x};\lambda)$  the Lax matrix (\ref{Lax:eq11}) and
 $D({\bf x})$ the deformation matrix (\ref{pema});
\item has as symmetry the bijection
\begin{align*}
\psi: \left(x^{(1)},x^{(2)},\ldots, x^{(N)}\right)\mapsto  \left(x^{(N)},x^{(1)},\ldots, x^{(N-1)}\right).
\end{align*}
\item The hierarchies $\mathcal{K}^{(0)}$ and $\mathcal{K}^{(1)}$ are not equivalent.
\end{enumerate}
\end{prop}

\begin{proof}
The compatibility conditions $\widehat L^{N,1}({\bf u};\lambda)\widehat L^{N,1}({\bf y};\lambda)=\widehat L^{N,1}({\bf v};\lambda)\widehat L^{N,1}({\bf x};\lambda)$ explicitly read:
\begin{align}
\begin{aligned}\label{sec3:2}
\left(1-u^{(i)}\right)^{-1} \left(1-y^{(i)}\right)^{-1}=\left(1-v^{(i)}\right)^{-1} \left(1-x^{(i)}\right)^{-1},
\end{aligned}\\
\begin{aligned}\label{sec3:3}
\left(1-u^{(i)}\right)^{-1}u^{(i)} \left(1-y^{(i-1)}\right)^{-1}y^{(i-1)}=\left(1-v^{(i)}\right)^{-1}v^{(i)} \left(1-x^{(i-1)}\right)^{-1}x^{(i-1)},
\end{aligned}\\
\begin{aligned}\label{sec3:4}
&\left(1-u^{(i)}\right)^{-1} \left(1-y^{(i)}\right)^{-1}y^{(i)}+\left(1-u^{(i)}\right)^{-1}u^{(i)} \left(1-y^{(i-1)}\right)^{-1}\\
&=\left(1-v^{(i)}\right)^{-1} \left(1-x^{(i)}\right)^{-1}x^{(i)}+\left(1-v^{(i)}\right)^{-1}v^{(i)} \left(1-x^{(i-1)}\right)^{-1}
\end{aligned}
\end{align}
where the superscript $i=1,\ldots, N,$ is considered modulo $N$. Similarly with the Proposition \ref{prop:01} the three sets of compatibility conditions presented above are not functionally independent since it can be easily shown that equations (\ref{sec3:4}) can be obtained by adding (\ref{sec3:2}) and (\ref{sec3:3}). Then it is easy to verify that the solved form of (\ref{sec3:2}),(\ref{sec3:3}), is exactly (\ref{sec3:1}). Furthermore, the compatibility conditions (\ref{sec3:2})-(\ref{sec3:4}) are clearly invariant under the map $\phi:=\psi\times \psi,$ that proves that $\psi$ is a symmetry of $\mathcal{K}^{(1)}.$
Finally,  the hierarchies $\mathcal{K}^{(1)}$ and $\mathcal{K}^{(0)}$ are related via
the change of variables
$\left(u^{(i)},x^{(i)},v^{(i)},y^{(i)}\right)\mapsto \left(u^{(i)},x^{(i+1)},v^{(i)},y^{(i+1)}\right),$ $\forall i\in\{1,\ldots, N\}.$  Indeed, under this change of variables,  $\mathcal{K}^{(1)}$ reads
\begin{align*}
\begin{aligned}
u^{(i)}=\left(1-y^{(i+1)}\right)^{-1}\left(x^{(i+1)}-y^{(i+1)}\right)x^{(i)}\left(x^{(i)}-y^{(i)}\right)^{-1}\left(1-y^{(i)}\right),\\
v^{(i)}=\left(1-x^{(i+1)}\right)^{-1}\left(y^{(i+1)}-x^{(i+1)}\right)y^{(i)}\left(y^{(i)}-x^{(i)}\right)^{-1}\left(1-x^{(i)}\right),\\
 i=1,2,\ldots, N,
 \end{aligned}
\end{align*}
that coincides with the {\em opposite} of the inverse hierarchy of maps $\mathcal{K}^{(0)}.$ Note that the {\em opposite} $H^{opp}$ of a non-Abelian map $H,$ is defined as the map that is obtained from the map $H$ when all multiplications are taken in opposite order.

Note also that there does not exists any bijection $\kappa$ such that  $\mathcal{K}^{(1)}=(\kappa^{-1}\times \kappa^{-1})\circ \mathcal{K}^{(0)} \circ (\kappa\times \kappa),$ so $\mathcal{K}^{(0)}$ and $\mathcal{K}^{(1)}$
are not equivalent up to the equivalence relation of Definition  \ref{def:1:1}.
\end{proof}

\subsubsection{Multidimensional compatibility of the hierarchies $\mathcal{K}^{(i)},$  $i=0,1$}

In this Section we prove that both hierarchies $\mathcal{K}^{(i)},$  $i=0,1$  are multidimensional compatible. Actually we prove the multidimensional compatibility of the $\mathcal{K}^{(1)}$ hierarchy by providing  explicitly its multidimensional compatibility formula. The multidimensional compatibility of  $\mathcal{K}^{(0)}$  can be proven  in an exactly similar manner.

By performing the identifications
\begin{align} \label{identif:2}
X^{i,a}:=x^{(i)},&& X_b^{i,a}:=u^{(i)}, &&X^{i,b}:=y^{(i)}, &&X_a^{i,b}:=v^{(i)},
\end{align}
 the hierarchy of maps (\ref{sec3:1}) obtains the compact form
\begin{align} \label{sec3:14}
X^{i,a}_b=\left(1-X^{i,b}\right)^{-1}\left(X^{i,a}-X^{i,b}\right) X^{i-1,a} \left(X^{i-1,a}-X^{i-1,b}\right)^{-1} \left(1-X^{i-1,b}\right),
\end{align}
$i=1,\ldots,N,$ $a\neq b\in\{1,2\}.$
   At the same time, (\ref{sec3:14}) serves as a difference system with variables assigned on the edges of an elementary cell of the $\mathbb{Z}^2$ graph, where the subscripts denote appropriate discrete shifts (see Figure \ref{fig1} under the identification (\ref{identif:2})).
\begin{lemma}\label{lemma1}
It holds
 \begin{align}  \label{Lax:lemma1:eq1}
1-X^{i,b}_c=\left(1-X^{i,c}\right)^{-1}\left(X^{i,b}-X^{i,c}\right)\Gamma^{(i)}(b,c),\\ \label{Lax:lemma1:eq2}
X^{i,a}_c-X^{i,b}_c=\begin{multlined}[t]
\left(1-X^{i,c}\right)^{-1} \left(X^{i,b}-X^{i,c}\right) \Delta^{(i)}(a,b,c) \left(X^{i-1,a}-X^{i-1,c}\right)^{-1} \\
\left(1-X^{i-1,c}\right),
\end{multlined}
 \end{align}
$i=1,\ldots,N,\; a\neq b \neq c \neq a\in\{1,\ldots,n\},$ where
 \begin{align*}
\Gamma^{(i)}(b,c):=\left(X^{i,b}-X^{i,c}\right)^{-1}\left(1-X^{i,c}\right)-X^{i-1,b}\left(X^{i-1,b}-X^{i-1,c}\right)^{-1}\left(1-X^{i-1,c}\right),\\
\Delta^{(i)}(a,b,c):=\begin{multlined}[t]
\left(X^{i,b}-X^{i,c}\right)^{-1}\left(X^{i,a}-X^{i,c}\right)X^{i-1,a}-X^{i-1,b}\left(X^{i-1,b}-X^{i-1,c}\right)^{-1}\\
\left(X^{i-1,a}-X^{i-1,c}\right).
\end{multlined}
 \end{align*}
The functions $\Gamma^{(i)}(b,c)$ and $\Delta^{(i)}(a,b,c)$ are antisymmetric under the interchange $b\leftrightarrow c$  i.e.
 $\Gamma^{(i)}(b,c)+\Gamma^{(i)}(c,b)=0,$   $\Delta^{(i)}(a,b,c)+\Delta^{(i)}(a,c,b)=0$.
\end{lemma}
\begin{proof}
Substituting from (\ref{sec3:14}) the expressions of $X^{i,a}_c$ and $X^{i,b}_c$ into the lhs of (\ref{Lax:lemma1:eq1}) and (\ref{Lax:lemma1:eq2}),  upon expansion, recollection of terms  we verify these formulae.

Using the identity
\begin{align*}
\left(1-A B^{-1}\right)^{-1}+\left(1-B A^{-1}\right)^{-1}=1,
\end{align*}
where $A,B$  non-commuting symbols, we can show that  $\Gamma^{(i)}(b,c)+\Gamma^{(i)}(c,b)=0,$   $\Delta^{(i)}(a,b,c)+\Delta{(i)}(a,c,b)=0$ and that proves the fact that the functions $\Gamma^{(i)}(b,c)$ and $\Delta^{(i)}(a,b,c),$ are antisymmetric under the interchange $b\leftrightarrow c$ of the discrete shifts.
\end{proof}
\begin{prop} \label{prop:3.4}
The hierarchy of difference equations (\ref{sec3:14}) can be extended in a compatible way to $n-$dimensions as follows
 \begin{align} \label{Lax:eq5}
X^{i,a}_b=\left(1-X^{i,b}\right)^{-1}\left(X^{i,a}-X^{i,b}\right) X^{i-1,a} \left(X^{i-1,a}-X^{i-1,b}\right)^{-1} \left(1-X^{i-1,b}\right),
\end{align}
with $i=1,\ldots,N,\; a\neq b\in\{1,\ldots,n\}$. 
   The compatibility conditions
 \begin{align*}
 X^{i,a}_{bc}=X^{i,a}_{cb},\quad i=1,\ldots,N,\;\; a\neq b \neq c \neq a\in\{1,\ldots,n\}
 \end{align*}
  hold.
\end{prop}
\begin{proof}
Shifting relations (\ref{Lax:eq5}) at the $c-$direction, we obtain
 \begin{align*}
X^{i,a}_{bc}=\left(1-X^{i,b}_c\right)^{-1}\left(X^{i,a}_c-X^{i,b}_c\right) X^{i-1,a}_c \left(X^{i-1,a}_c-X^{i-1,b}_c\right)^{-1} \left(1-X^{i-1,b}_c\right).
 \end{align*}
Substituting from (\ref{Lax:eq5}) the expression of $X^{i-1,a}_c,$ $X^{i-1,b}_c,$ $X^{i,a}_c$ and $X^{i,b}_c$   to the relations above
and by making use of Lemma (\ref{lemma1}) we obtain the following multidimensional
 compatibility formula
\begin{align*}
 X^{i,a}_{bc}=\left(\Gamma^{(i)}(b,c)\right)^{-1}\Delta^{(i)}(a,b,c)X^{i-2,a}\left(\Delta^{(i-1)}(a,b,c)\right)^{-1}\Gamma^{(i-1)}(b,c),
 \end{align*}
that is clearly symmetric under the interchange  $b$ to $c,$ since it consists of the product of an even number of antisymmetric functions (under the interchange $b\leftrightarrow c$) and that completes the proof.
\end{proof}
\begin{remark}
In exactly similar manner we can prove that the hierarchy of difference equations (\ref{sec3:01}) can be extended in a compatible way to $n-$dimensions as follows
 \begin{align} \label{K-hi}
X^{i,a}_b=\left(1-X^{i-1,b}\right)^{-1}\left(X^{i-1,a}-X^{i-1,b}\right) X^{i,a} \left(X^{i,a}-X^{i,b}\right)^{-1} \left(1-X^{i,b}\right),
\end{align}
with $i=1,\ldots,N,\; a\neq b\in\{1,\ldots,n\}$. If we consider $i\in\mathbb{Z},$ (\ref{K-hi}) remains multidimensional compatible and moreover it serves as a deformation of the Hirota-Miwa system. Indeed by setting $X^{k,d}=\epsilon X^{k,d},$ $k\in \mathbb{Z},$ $d\in\{1,\ldots,n\},$ $\epsilon \in C(\mathbb{A}^\times),$ and by taking the limit $\epsilon\rightarrow 0,$  we obtain exactly the  non-Abelian discrete KP hierarchy, which was derived in \cite{Doliwa_2014}, from the non-Abelian Hirota-Miwa system \cite{Nimmo:2006}. In that respect and for $i \in \mathbb{Z},$    (\ref{K-hi}) serves as an integrable  deformation of the discrete KP hierarchy.

\end{remark}

\subsubsection{Quadrirationality of the hierarchies $\mathcal{K}^{(i)}$  $i=0,1$}

In what follows, we consider an analysis of proving quadrirationality of the $\mathcal{K}^{(1)}$ hierarchy under some commutativity assumptions that are referred to as the {\em centrality assumptions}.

\begin{remark} \label{lemma:3.20}
The products $\mathcal{P}^{(N,k)}({\bf x}):=\prod_{l=1}^N x^{(N+k-l)},$ 
$k=1,2,\ldots,N$ satisfy the relations:
\begin{align} \label{Lax:pr1}
\mathcal{P}^{(N,k)}({\bf u})=\begin{multlined}[t]
\left(1-y^{(N+k-1)}\right)^{-1}\left(x^{(N+k-1)}-y^{(N+k-1)}\right)\mathcal{P}^{N,k-1}({\bf x})\\
\left(x^{(N+k-1)}-y^{(N+k-1)}\right)^{-1}\left(1-y^{(N+k-1)}\right),
\end{multlined}\\ \label{Lax:pr2}
\mathcal{P}^{(N,k)}({\bf v})=\begin{multlined}[t]
\left(1-x^{(N+k-1)}\right)^{-1}\left(y^{(N+k-1)}-x^{(N+k-1)}\right)\mathcal{P}^{N,k-1}({\bf y})\\
\left(y^{(N+k-1)}-x^{(N+k-1)}\right)\left(1-x^{(N+k-1)}\right),
\end{multlined}
 \end{align}
$k=1,\ldots,N,$ and  ${\bf u},{\bf v}$ stand for  the defining relations  of the $\mathcal{K}^{(1)}$ hierarchy (\ref{sec3:1}). Note that the direct substitution of (\ref{sec3:1}) into (\ref{Lax:pr1}),(\ref{Lax:pr2}) validates the formulas above.  Also note that  we consider the products presented above in ascending order of the superscripts f.i. $\prod_{l=1}^m x^{(l)}:=x^{(1)}x^{(2)}\ldots x^{(m)}.$
\end{remark}

Clearly in the Abelian case the products $\mathcal{P}^{(N,k)}$ are independent of $k$ i.e. all products $\mathcal{P}^{(N,k)}$ coincide with the product   $\mathcal{P}^{(N)}({\bf x})=\prod_{l=1}^N x^{(l)}.$ Furthermore,  there is $\mathcal{P}^{(N)}({\bf u})=\mathcal{P}^{(N)}({\bf x}),$ and $\mathcal{P}^{(N)}({\bf v})=\mathcal{P}^{(N)}({\bf y}),$ hence the products $\mathcal{P}^{(N)}({\bf x})$ and $\mathcal{P}^{(N)}({\bf y})$ serve as invariants of the map $\mathcal{K}^{(1)}.$ In the non-Abelian case, these facts mentioned earlier are no longer true. Due to  Remark \ref{lemma:3.20}, the products $\mathcal{P}^{(N,k)}$ are no longer independent of $k$ and moreover are not invariants of the map. Nevertheless,  if we assume that the products $\mathcal{P}^{(N,k)}$ belong to the center of the algebra $\mathbb{A}^\times,$  from  Remark \ref{lemma:3.20} we obtain the invariant relations $\mathcal{P}^{(N,k)}({\bf u})=\mathcal{P}^{(N,k-1)}({\bf x}),$ and
$\mathcal{P}^{(N,k)}({\bf v})=\mathcal{P}^{(N,k-1)}({\bf y}),$ $k=1,\ldots,N,$ hence the functions  $\mathcal{P}^{(N,k)}({\bf x}),\mathcal{P}^{(N,k)}({\bf y})$ are covariants of the map. We denote the functions $\mathcal{P}^{(N,k)}({\bf x}),\mathcal{P}^{(N,k)}({\bf y}),$ respectively as $\mathrm{p}^{(k-1)}$ and $\mathrm{q}^{(k-1)}$. To recapitulate we have
\begin{align} \label{sec3:20}
\mathrm{p}^{(k-1)}:=\mathcal{P}^{(N,k)}({\bf x}), && \mathrm{q}^{(k-1)}:=\mathcal{P}^{(N,k)}({\bf y}), && i=1,\ldots,N.
\end{align}
In addition, as a consequence of (\ref{sec3:20})  it holds that $\mathrm{p}^{(i)}=\mathrm{p}^{(j)},$ and  $\mathrm{q}^{(i)}=\mathrm{q}^{(j)},$   $\forall i,j\in \{0,1,\ldots,N-1\}.$

  From further on, when we refer to the centrality assumption, we refer to the formulas:
\begin{align} \label{sec3:5}
x^{(N)}\cdots x^{(2)} x^{(1)}=p\in C(\mathbb{A}^\times), && y^{(N)}\cdots y^{(2)} y^{(1)}=q\in C(\mathbb{A}^\times),
\end{align}
where with $C(\mathbb{\mathbb{A}^\times})$ we denote the center of the algebra $\mathbb{A}^\times$ and we also have denoted $p:=\mathrm{p}^{(0)},$ $q:=\mathrm{q}^{(0)}.$
The centrality assumptions were first introduced in \cite{Doliwa_2013,Doliwa_2014} for the so-called $N-$periodic reduction of the non-Abelian Hirota-Miwa system (KP-map). Centrality assumptions play  a  crucial role to the quadrirationality of the hierarchy of  maps (\ref{sec3:1}), as it is shown in the Proposition that follows.

\begin{prop}\label{prop1.2}
The hierarchy of maps (\ref{sec3:1}) is birational.  When the centrality assumptions (\ref{sec3:5}) are imposed the hierarchy of maps is quadrirational.
\end{prop}
\begin{proof}
First we prove that  (\ref{sec3:1}) is birational. The compatibility conditions (\ref{sec3:2}), (\ref{sec3:3}),  can be solved rationally for $x^{(i)},y^{(i)}$ in terms of $u^{(i)},v^{(i)}$ and that proves biratonality  of the system. Specifically from (\ref{sec3:2}), (\ref{sec3:3}) we obtain:
\begin{align} \label{sec3:6}
\begin{aligned}
x^{(i)}=\left(1-v^{(i+1)}\right)\left(u^{(i+1)}-v^{(i+1)}\right)^{-1}u^{(i+1)}\left(u^{(i)}-v^{(i)}\right)\left(1-v^{(i)}\right)^{-1},\\
y^{(i)}=\left(1-u^{(i+1)}\right)\left(v^{(i+1)}-u^{(i+1)}\right)^{-1}v^{(i+1)}\left(v^{(i)}-u^{(i)}\right)\left(1-u^{(i)}\right)^{-1},\\
i=1,2,\ldots, N.
\end{aligned}
\end{align}
From (\ref{sec3:6}) it is easy to validate the following formulae
\begin{align*} 
\begin{aligned}
\mathcal{P}^{(N,k)}({\bf x})=\left(1-v^{(N+k)}\right)^{-1}\left(u^{(N+k)}-v^{(N+k)}\right)\mathcal{P}^{(N,k+1)}({\bf u})\left(u^{(N+k)}-v^{(N+k)}\right)^{-1}\left(1-v^{(N+k)}\right),\\
\mathcal{P}^{(N,k)}({\bf y})=\left(1-u^{(N+k)}\right)^{-1}\left(v^{(N+k)}-u^{(N+k)}\right)\mathcal{P}^{(N,k+1)}({\bf v})\left(v^{(N+k)}-u^{(N+k)}\right)^{-1}\left(1-u^{(N+k)}\right),\\
k=1,\ldots,N,
\end{aligned} &&
\end{align*}
where the expressions $\mathcal{P}^{(N,k)}$  are defined in  Remark \ref{lemma:3.20}.

To prove that (\ref{sec3:1}) is quadrirational it suffices to solve (\ref{sec3:1})  rationally for $u^{(i)},y^{(i)}$ in terms of $x^{(i)},v^{(i)}$  and show that the resulting map is birational.

From (\ref{sec3:2}) and  (\ref{sec3:3}) we obtain respectively
 \begin{align} \label{sec3:7}
\left(1-u^{(i)}\right)^{-1}=\left(1-v^{(i)}\right)^{-1} \left(1-x^{(i)}\right)^{-1} \left(1-y^{(i)}\right),
\end{align}
and
 \begin{align} \label{sec3:8}
\left(1-u^{(i)}\right)^{-1}u^{(i)}=\left(1-v^{(i)}\right)^{-1}v^{(i)} \left(1-x^{(i-1)}\right)^{-1}x^{(i-1)}\left(y^{(i-1)}\right)^{-1}\left(1-y^{(i-1)}\right).
\end{align}
 Substituting these expressions into the lhs of (\ref{sec3:4}) we get
 \begin{align} \label{sec3:9}
y^{(i)}y^{(i-1)}-\left(x^{(i)}+A^{(i)}\right)y^{(i-1)}+A^{(i)}x^{(i-1)}=0,
\end{align}
where
\begin{align*}
A^{(i)}:=\left(1-x^{(i)}\right)v^{(i)}\left(1-x^{(i-1)}\right)^{-1}.
\end{align*}
Multiplying from the right equation (\ref{sec3:9}) respectively with
$$
y^{(i-2)},\quad y^{(i-2)} y^{(i-3)},\quad \ldots, \quad \prod_{l=2}^{m}y^{(i-l)},
$$
 and by substituting at each step from (\ref{sec3:9})  the expressions $y^{(k)}y^{(k-1)}$ we obtain:
\begin{align} \label{sec3:10}
\prod_{l=0}^m y^{(i-l)}-f^{(i)}_m y^{(i-m)}+f^{(i)}_{m-1} A^{(i-m+1)}x^{(i-m)}=0,&& m\geq 2,
\end{align}
where $f^{(i)}$ satisfies the recurrences
\begin{align} \label{sec3:11}
f^{(i)}_{n+2}=f^{(i)}_{n+1}\left(x^{(i-n-1)}+A^{(i-n-1)}\right)-f^{(i)}_n A^{(i-n)}x^{(i-n-1)},&& n\in\mathbb{Z},
\end{align}
with
\begin{align} \label{sec3:12}
f^{(i)}_0=1,&&f^{(i)}_1=x^{(i)}+A^{(i)}, && i=1,\ldots, N.
\end{align}

Setting $m=N-1$ and by assuming the centrality assumptions (\ref{sec3:5}), equations (\ref{sec3:10}) read
\begin{align} \label{sec3:13}
q-f^{(i)}_{N-1} y^{(i-N+1)}+f^{(i)}_{N-2} A^{(i-N+2)}x^{(i-N+1)}=0,
\end{align}
where $f^{(i)}_{N-1},f^{(i)}_{N-2}$ are determined by the recurrences (\ref{sec3:11}),(\ref{sec3:12}). So from  (\ref{sec3:13}) we have obtained $y^{(i)}$ as a function of $x^{(i)},v^{(i)},$ $i=1,\ldots, N,$ and together with (\ref{sec3:7}) we finally obtain $y^{(i)},u^{(i)},$ in terms of $v^{(i)},x^{(i)},$ i.e. the companion hierarchy of  maps of the hierarchy of  maps $\mathcal{K}^{(1)}$. In exactly similar manner we can express rationally   $x^{(i)},v^{(i)}$ in terms of $u^{(i)},y^{(i)},$ that proves birationality of the companion hierarchy of maps, and that completes the proof.
\end{proof}
%
\begin{prop} \label{prop3:YB}
 Let the expressions  $f^{(i)}_{N-1},f^{(i)}_{N-2},$  $i=1,\ldots, N,$ be determined by the recurrences, (\ref{sec3:11}),(\ref{sec3:12}), and
the expressions  $g^{(i)}_{N-1},g^{(i)}_{N-2},$  be  determined by the recurrences
\begin{align} \label{rec3:11}
g^{(i)}_{n+2}=\left(v^{(i+n)}+B^{(i+n+1)}\right)g^{(i)}_{n+1}-v^{(i+n)}B^{(i+n)}g^{(i)}_n,&& n\in\mathbb{Z},
\end{align}
with
\begin{align} \label{recc3:12}
g^{(i)}_0=1,&&g^{(i)}_1=v^{(i-1)}+B^{(i)}, && B^{(i)}:=\left(1-v^{(i)}\right)^{-1}x^{(i-1)}\left(1-v^{(i-1)}\right),&&i=1,\ldots, N.
\end{align}
Assuming the centrality assumptions (\ref{sec3:5}) that results quadrirationality,  the companion hierarchy $c\mathcal{K}^{(1)}$ of the hierarchy of maps $\mathcal{K}^{(1)}$ explicitly reads
\begin{align*}
c\mathcal{K}^{(1)}: (x^{(1)},\ldots, x^{(N)},v^{(1)},\ldots, v^{(N)})\mapsto  (u^{(1)},\ldots, u^{(N)},y^{(1)},\ldots, y^{(N)}),
\end{align*}
where
\begin{align} \label{sec3:1000}
\begin{aligned}
u^{(i)}=\left(p+v^{(i)}B^{(i)}g^{(i-N+1)}_{N-2}\right)\left(g^{(i-N+1)}_{N-1}\right)^{-1},\\
y^{(i)}=\left(f^{(i+N-1)}_{N-1}\right)^{-1} \left(q+f^{(i+N-1)}_{N-2}A^{(i+1)}x^{(i)}\right),
\end{aligned}& & i=1,\ldots, N,
\end{align}
and it serves as a hierarchy  of Yang-Baxter maps.
\end{prop}
\begin{proof}
The formulas of $y^{(i)}$ in (\ref{sec3:1000}) are just the solved for $y^{(i)}$ form  of (\ref{sec3:13}).  To obtain the formulas for $u^{(i)}$ of (\ref{sec3:1000}), so that we  obtain the explicit formulas for the hierarchy $c\mathcal{K}^{(1)}$, we have to substitute the formulas of $y^{(i)}$  to  (\ref{sec3:7}) and solve for $u^i.$ Equivalently, from the defining relations of $\mathcal{K}^{(1)}$ (\ref{sec3:2}) and  (\ref{sec3:3}), if we eliminate $y^{(i)}$ we will obtain $u^{(i)}$ as  functions  of $v^{(i)}$ and $x^{(i)},$ hence the first part of the formulas (\ref{sec3:1000}). This is what we do for the rest of the proof.

From (\ref{sec3:2}) and  (\ref{sec3:3}) we obtain
 \begin{align} \label{sec30:9}
u^{(i+1)}u^{(i)}-u^{(i+1)}\left(v^{(i)}+B^{(i+1)}\right)+v^{(i+1)}B^{(i+1)}=0,
\end{align}
where
\begin{align*}
B^{(i)}:=\left(1-v^{(i)}\right)^{-1}x^{(i-1)}\left(1-v^{(i-1)}\right).
\end{align*}
Multiplying from the left equation (\ref{sec30:9}) respectively with
$$
u^{(i+2)},\quad u^{(i+3)} y^{(i+2)},\quad \ldots,\quad \prod_{l=0}^{m-2}u^{(i+m-l)}
$$
 and by substituting at each step from (\ref{sec30:9})  the expressions $u^{(k+1)}u^{(k)}$ we obtain:
\begin{align} \label{sec30:10}
\prod_{l=0}^m u^{(i+m-l)}-u^{(i+m)}g^{(i)}_m+v^{(i+m)}B^{(i+m)}g^{(i)}_{m-1}=0,&& m\geq 2,
\end{align}
where $g^{(i)}$ satisfies the recurrences
\begin{align} \label{sec30:11}
g^{(i)}_{n+2}=\left(v^{(i+n+1)}+B^{(i+n+2)}\right)g^{(i)}_{n+1}-v^{(i+n+1)}B^{(i+n+1)}g^{(i)}_n,&& n\in\mathbb{Z},
\end{align}
with
\begin{align} \label{sec30:12}
g^{(i)}_0=1,&&g^{(i)}_1=v^{(i-1)}+B^{(i)}, && i=1,\ldots, N.
\end{align}

Setting $m=N-1$ and by assuming the centrality assumption (\ref{sec3:5}), equations (\ref{sec30:10}) read
\begin{align} \label{sec30:13}
p-u^{(i+N-1)}g^{(i)}_{N-1}+v^{(i+N-1)}B^{(i+N-1)}g^{(i)}_{N-2}=0,
\end{align}
where $g^{(i)}_{N-1},g^{(i)}_{N-2}$ are determined by the recurrences (\ref{sec30:11}),(\ref{sec30:12}). Solving (\ref{sec30:13}) for $u^{(i)},$ we obtain  the first expression of (\ref{sec3:1000}). Since $c\mathcal{K}^{(1)}$ serves as the companion hierarchy of a $3D-$compatible hierarchy,  $c\mathcal{K}^{(1)}$ is hierarchy of Yang-Baxter maps (see Section \ref{Section:2})  and that completes the proof.

\end{proof}

\begin{remark}
By following an exactly similar analysis as in this Section, we can show the quadrirationality of the hierarchy  $\mathcal{K}^{(0)}$ and provide its companion that serves as a hierarchy of Yang-Baxter maps.
\end{remark}

\section{Hierarchies of integrable difference systems in vertex variables} \label{Section:4}

The hierarchies of integrable difference systems in edge variables  associated with the  $3D$-compatible maps $\mathcal{K}^{(0)}$ and $\mathcal{K}^{(1)},$ can be also rewritten in terms of vertex variables. The procedure that incorporates the transition  from edge (or even face) to vertex variables in integrable difference systems, was introduced in \cite{doliwa-santini}. This procedure is widely applied nowadays \cite{Kassotakis_2011,Kassotakis_2012,Doliwa_2013,Kassotakis_2018,Fordy_2017,Kassotakis_2021}, under the attributed  name {\em potentialisation}.

In what follows we apply the potentialization procedure to the integrable hierarchy $\mathcal{K}^{(0)}$, to obtain integrable hierarchies of difference systems in vertex variables. Applying the potentialization procedure to the hierarchy $\mathcal{K}^{(1)},$  leads to point equivalent integrable hierarchies of difference systems with the hierarchy $\mathcal{K}^{(0)}$.

From the defining relations of the $3D-$compatible map $\mathcal{K}^{(0)},$ or equivalently from the compatibility conditions (\ref{sec3:2}),(\ref{sec3:3}), it is guaranteed the existence of two sets of potential functions that allow us to rewrite the $3D-$compatible map as  two (related by a B\"acklund transformation)  integrable hierarchies of difference systems with dynamical variables defined on the vertices of the $\mathbb{Z}^2$ graph.

\subsection{The first set of potential functions and the lattice-modified Gel'fand-Dikii  hierarchy}

The  integrable hierarchy $\mathcal{K}^{(0)},$ in its polynomial form consists of the  sets of equations (\ref{sec3:02}) and (\ref{sec3:03}). The set of equations (\ref{sec3:02}), guarantees the existence of potential functions $\phi^{(i)},$ $i=1,\ldots,N,$  such that
\begin{align}\label{sec4:1}
\begin{aligned}
\left(1-x^{(i)}\right)^{-1}=\phi^{(i)}_1\left(\phi^{(i-1)}\right)^{-1}, & \left(1-y^{(i)}\right)^{-1}=\phi^{(i)}_2\left(\phi^{(i-1)}\right)^{-1},\\
\left(1-u^{(i)}\right)^{-1}=\phi^{(i)}_{12}\left(\phi^{(i-1)}_2\right)^{-1}, & \left(1-v^{(i)}\right)^{-1}=\phi^{(i)}_{12}\left(\phi^{(i-1)}_1\right)^{-1},
\end{aligned}&&  i=1,\ldots,N.
\end{align}
In terms of the potential functions $\phi^i,$ the set of equations (\ref{sec3:02}) is identically satisfied, while the set of equations (\ref{sec3:03}) becomes
\begin{align} \label{sec4:2}
\fl\begin{multlined}[t]\left(\phi^{(i)}_{12}\left(\phi^{(i-1)}_2\right)^{-1}-1\right)\left(\phi^{(i)}_{2}\left(\phi^{(i-1)}\right)^{-1}-1\right)\\
=\left(\phi^{(i)}_{12}\left(\phi^{(i-1)}_1\right)^{-1}-1\right)
\left(\phi^{(i)}_{1}\left(\phi^{(i-1)}\right)^{-1}-1\right),
\end{multlined}
\end{align}
$i=1,\ldots,N$ and constitute a hierarchy of difference systems in vertex variables.
\begin{prop} \label{prop4:1}
For the hierarchy of difference systems in vertex variables (\ref{sec4:2}) we have:
\begin{enumerate}
\item  it arises as the compatibility conditions of the Lax equation
\begin{align*}
L(\bff{ \phi}_{12},{\bff \phi}_2;\lambda)L({\bff \phi}_{2},{\bff \phi};\lambda)=L({\bff \phi}_{12},{\bff \phi}_1;\lambda)L({\bff \phi}_{1},{\bff \phi};\lambda),
\end{align*}
 associated with the strong Lax matrix
\begin{align*}
L={\scriptsize\begin{pmatrix}
                      \phi^{(N)}_1\left(\phi^{(N-1)}\right)^{-1}-1&0&\cdots&0&\lambda\,\phi^{(N)}_1\left(\phi^{(N-1)}\right)^{-1}\\
                      \phi^{(1)}_1\left(\phi^{(N)}\right)^{-1}&\phi^{(1)}_1\left(\phi^{(N)}\right)^{-1}-1& 0& \cdots &0 \\
                      0  &\phi^{(2)}_1\left(\phi^{(1)}\right)^{-1}&\ddots& {}  &\vdots\\
                      \vdots & &\ddots &\phi^{(N-2)}_1\left(\phi^{(N-1)}\right)^{-1}-1 &0\\
                       0     & 0 &   &  \phi^{(N-1)}_1\left(\phi^{(N-2)}\right)^{-1}    &\phi^{(N-1)}_1\left(\phi^{(N-2)}\right)^{-1}-1
                     \end{pmatrix};}
\end{align*}
\item it is multidimensional consistent;
\item it  respects the {\it rhombic} symmetry
\begin{align*}
\tau :(\bff{\phi},{\bff \phi}_1,{\bff \phi}_2,{\bff \phi}_{12})\mapsto ({\bff \phi},{\bff \phi}_2,{\bff \phi}_1,{\bff \phi}_{12}),&&
\sigma : ({\bff \phi},{\bff \phi}_1,{\bff \phi}_2,{\bff \phi}_{12})\mapsto ({\bff \phi}_{12},{\bff \phi}_2,{\bff \phi}_1,{\bff \phi});
\end{align*}
\item it is an integrable hierarchy of difference systems in vertex variables.
\end{enumerate}
\end{prop}
\begin{proof}
Let us prove the statements of this Proposition.
\begin{enumerate}
\item We substitute the expressions of the potential functions (\ref{sec4:1}) into the Lax matrix $\widehat M^{(N,1)}({\bf x};\lambda)$ of Proposition (\ref{prop:01}) to obtain the Lax matrix presented here.
\item The multidimensional consistency of (\ref{sec4:2})  is a direct consequence of the multidimensional compatibility  of the underlying difference system in edge variables (\ref{sec3:01}). Indeed, for the system  (\ref{sec3:01})
      that its multidimensional extension reads
    \begin{align} \label{K1:eqs}
    X^{i,a}_b=\left(1-X^{i-1,b}\right)^{-1}\left(X^{i-1,a}-X^{i-1,b}\right)X^{i,a}\left(X^{i,a}-X^{i,b}\right)^{-1}\left(1-x^{i,b}\right),
    \end{align}
    $i=1,\ldots,N,\; a\neq b \neq c \neq a\in\{1,\ldots,n\},$
     it holds the multidimensional compatibility formula
    \begin{align}\label{sec4:3}
 X^{i,a}_{bc}=X^{i,a}_{cb},\quad i=1,\ldots,N,\;\; a\neq b \neq c \neq a\in\{1,\ldots,n\}.
 \end{align}
  Also for the potential functions $\phi^{(i)}$ we have
        \begin{align} \label{sec4:4}
  \left(1-X^{i,a}\right)^{-1}=\phi^{(i)}_a(\phi^{(i-1)})^{-1}, \quad i=1,\ldots,N, \quad a\in \{1,\ldots, n\}.
    \end{align}
From (\ref{sec4:3}) and (\ref{sec4:4}) we obtain
\begin{align*}
\phi^{(i)}_{abc}\left(\phi^{(i-1)}_{bc}\right)^{-1}=\phi^{(i)}_{acb}\left(\phi^{(i-1)}_{cb}\right)^{-1}.
\end{align*}
So in order to prove the multidimensional consistency of (\ref{sec4:2}) that is
\begin{align*}
\phi^{(i)}_{abc}=\phi^{(i)}_{acb}, \quad i=1,\ldots,N,\;\; a\neq b \neq c \neq a\in\{1,\ldots,n\},
\end{align*}
it suffices to prove that
\begin{align} \label{pformula}
 \phi^{(i)}_{ab}=\phi^{(i)}_{ba},&& a\neq b \in\{1,\ldots,n\}.
 \end{align}
Indeed  by shifting appropriately  (\ref{sec4:4}) we obtain
\begin{align*}
\left(1-X^{i,a}_b\right)^{-1}\phi^{(i-1)}_{b}=\phi^{(i)}_{ab},&&\left(1-X^{i,b}_a\right)^{-1}\phi^{(i-1)}_{a}=\phi^{(i)}_{ba}
  \end{align*}
  and by using again (\ref{sec4:4}) to eliminate $\phi^{(i-1)}_{b}$ and $\phi^{(i-1)}_{a},$ we arrive to
\begin{align} \label{fp:000}
\begin{aligned}
\left(1-X^{i,a}_b\right)^{-1}\left(1-X^{i-1,b}\right)^{-1}\phi^{(i-2)}=\phi^{(i)}_{ab},\\
\left(1-X^{i,b}_a\right)^{-1}\left(1-X^{i-1,a}\right)^{-1}\phi^{(i-2)}=\phi^{(i)}_{ba}.
\end{aligned}
  \end{align}
From (\ref{K1:eqs}) we get
\begin{align} \label{pppw}
\begin{aligned}
1-X^{i,a}_b=\left(1-X^{i-1,b}\right)^{-1}\left(X^{i-1,a}-X^{i-1,b}\right)\widehat \Gamma^{(i)}(a,b),\\
1-X^{i,b}_a=\left(1-X^{i-1,a}\right)^{-1}\left(X^{i-1,b}-X^{i-1,a}\right)\widehat \Gamma^{(i)}(b,a),
\end{aligned}
\end{align}
where the expressions $\Gamma^{(i)}(a,b)$ are antisymmetric under the interchange $a\leftrightarrow b,$ and they are of  similar form to the ones of Lemma \ref{lemma1}. Using (\ref{pppw}) to eliminate $1-X^{i,a}_b$ and $1-X^{i,b}_a$ from (\ref{fp:000}), together with  the fact that $\Gamma^{(i)}(a,b)+\Gamma^{(i)}(b,a)=0,$ we obtain exactly (\ref{pformula})
and that proves the multidimensional consistency of  the hierarchy of difference systems in vertex variables (\ref{sec4:2}).
\item It can be easily verified that (\ref{sec4:2})  is invariant under $\tau,$ while acting with $\sigma$ on (\ref{sec4:2}) we get
\begin{align} \label{sec4:22}
\begin{multlined}[t]
\left(\phi^{(i)}\left(\phi^{(i-1)}_1\right)^{-1}-1\right)\left(\phi^{(i)}_{1}\left(\phi^{(i-1)}_{12}\right)^{-1}-1\right)\\
=\left(\phi^{(i)}\left(\phi^{(i-1)}_2\right)^{-1}-1\right)\left(\phi^{(i)}_{2}\left(\phi^{(i-1)}_{12}\right)^{-1}-1\right),
\end{multlined}
\end{align}
which is (\ref{sec4:2}) in disguise.
  Indeed, by acting on (\ref{sec4:2}) with  ${\bf T}_{-1} {\bf T}_{-2}$ we obtain
\begin{align*}
\begin{multlined}[t]
\left(\phi^{(i)}\left(\phi^{(i-1)}_{-1}\right)^{-1}-1\right)\left(\phi^{(i)}_{-1}\left(\phi^{(i-1)}_{-1-2}\right)^{-1}-1\right)\\
=\left(\phi^{(i)}\left(\phi^{(i-1)}_{-2}\right)^{-1}-1\right)\left(\phi^{(i)}_{-2}\left(\phi^{(i-1)}_{-1-2}\right)^{-1}-1\right),
\end{multlined}
\end{align*}
and if we perform the  change of the dependent variables $\phi^{(i)}_{m,n}= \Phi^{(i)}_{-m,-n},$ followed by the change of independent variables $m'=-m,n'=-n,$ we obtain that $\Phi^{(i)}_{m',n'}$ satisfies
(\ref{sec4:22}). Here with ${\bf T}_j$ we denote the forward shift operator in the $j-$th direction while with ${\bf T}_{-k}$ we denote the backward shift operator in the $k-$th direction i.e.
\begin{align*}
{\bf T}_1: \phi^{(i)}\mapsto  \phi^{(i)}_1,\;\; {\bf T}_{-1}: \phi^{(i)}\mapsto  \phi^{(i)}_{-1},\;\; {\bf T}_{-2}: \phi^{(i)}\mapsto  \phi^{(i)}_{-2},\;\; \mbox{etc}.
\end{align*}
\item Due to statements $(1)-(3)$, (\ref{sec4:2})  constitutes an integrable hierarchy of difference systems in vertex variables defined on the black-white (chessboard) lattice \cite{Papageorgiou:2009II,ABS:2009,Boll:2011}.
\end{enumerate}
\end{proof}

Note that the hierarchy of difference systems (\ref{sec4:2}), can be solved rationally only for the sets of variables  ${\bff \phi}$ and ${\bff \phi}_{12}$. If the centrality assumptions (\ref{sec3:5}) are imposed, that in terms of the potential functions ${\bff \phi}$ read
\begin{align} \label{sec4:5}
\prod_{l=1}^N\left(1-\phi^{(N-l)}\left(\phi^{(N+1-l)}_1\right)^{-1}\right)=p\in C(\mathbb{A}^\times),\\ \label{sec44:5} \prod_{l=1}^N\left(1-\phi^{(N-l)}\left(\phi^{(N+1-l)}_2\right)^{-1}\right)=q\in C(\mathbb{A}^\times),
\end{align}
 then (\ref{sec4:2}) can be solved  for any corner variable set, mimicking the  Abelian case.

 To recapitulate, hierarchy (\ref{sec4:2}) serves as an integrable hierarchy in vertex variables that respects the rhombic symmetry. If in addition the centrality assumptions (\ref{sec4:5}),(\ref{sec44:5}) are imposed, then one potential function can be eliminated and the resulting hierarchy consists of $N-1$ equations in $N-1$ potential functions. The resulting hierarchy  serves as the non-Abelian lattice-modified Gel'fand-Dikii  hierarchy, since its lowest member turns out to be the lattice-potential-modified  KdV equation  (in the abelian case derived in \cite{Bianchi:1894,Hir-sG,nij-qui-cap} and referred to as $(H3)_0$ in \cite{ABS}),   as the example that follows suggest. Remember that  the lattice-modified Gel'fand-Dikii  hierarchy in the Abelian case was firstly considered implicitly in \cite{Nijhoff:1992} and explicitly in \cite{Atkinson:2012}, whereas in the non-Abelian setting in \cite{Doliwa_2013}.
 \begin{example}[$N=2$]
  The first member of the hierarchy (\ref{sec4:2}) ($N=2$)  reads:
 \begin{align}\label{mkdv}
 \begin{multlined}[t]
\left(\phi^{(1)}_{12}\left(\phi^{(2)}_2\right)^{-1}-1\right)\left(\phi^{(1)}_{2}\left(\phi^{(2)}\right)^{-1}-1\right)
=\left(\phi^{(1)}_{12}\left(\phi^{(2)}_1\right)^{-1}-1\right)\\
\left(\phi^{(1)}_{1}\left(\phi^{(2)}\right)^{-1}-1\right),
\end{multlined} \\ \label{mkdv1111}
\begin{multlined}[t]
\left(\phi^{(2)}_{12}\left(\phi^{(1)}_2\right)^{-1}-1\right)\left(\phi^{(2)}_{2}\left(\phi^{(1)}\right)^{-1}-1\right)
=\left(\phi^{(2)}_{12}\left(\phi^{(1)}_1\right)^{-1}-1\right)\\
\left(\phi^{(2)}_{1}\left(\phi^{(1)}\right)^{-1}-1\right).
\end{multlined}
  \end{align}
  If we impose the centrality assumptions (\ref{sec4:5}),(\ref{sec44:5}) that now read
  \begin{align*}
  \begin{aligned}
  \left(1-\phi^{(1)}  \left(\phi^{(2)}_1\right)^{-1}\right)  \left(1-\phi^{(2)}  \left(\phi^{(1)}_1\right)^{-1}\right)=p,\\
     \left(1-\phi^{(1)}  \left(\phi^{(2)}_2\right)^{-1}\right)  \left(1-\phi^{(2)}  \left(\phi^{(1)}_2\right)^{-1}\right)=q,
     \end{aligned}
  \end{align*}
  we can eliminate f.i. $\phi^2$ and its shifts from  (\ref{mkdv}) to obtain
  \begin{align*}
  \left(1-q-\phi\, \phi_{12}^{-1}\right)\phi_2=  \left(1-p-\phi\, \phi_{12}^{-1}\right)\phi_1,
  \end{align*}
  where for simplicity we have denoted $\phi:=\phi^{(1)},$ $\phi_2:=\phi^{(1)}_2,$ etc. Under the re-parametrization $p\mapsto 1-1/p^2,$ $q\mapsto 1-1/q^2,$ followed by the point transformation $\phi_{m+m_1,n+n_1}\mapsto p^{m_1} q^{n_1} \phi_{m+m_1,n+n_1},$ the equation above takes the form
  \begin{align} \label{MKDV}
    \left(\frac{1}{q}-\frac{1}{p}\phi\, \phi_{12}^{-1}\right)\phi_2=\left(\frac{1}{p}-\frac{1}{q}\phi\, \phi_{12}^{-1}\right)\phi_1,
  \end{align}
  that is exactly $(H3)_0$ but extended in the non-Abelian domain.
 \end{example}

\subsection{The second set of potential functions and the lattice-$(Q3)_0$ Gel'fand-Dikii  hierarchy}

As we mentioned earlier, the   polynomial form of the integrable hierarchy $\mathcal{K}^{(0)},$ consists of the  sets of equations (\ref{sec3:02}) and (\ref{sec3:03}). The set of equations (\ref{sec3:03}), guarantees the existence of potential functions $\psi^i,$ $i=1,\ldots,N,$  such that
\begin{align}\label{sec4:11}
\begin{aligned}
\left(1-x^{(i)}\right)^{-1}x^{(i)}=\psi^{(i)}_1\left(\psi^{(i)}\right)^{-1}, & \left(1-y^{(i)}\right)^{-1}y^{(i)}=\psi^{(i)}_2\left(\psi^{i}\right)^{-1},\\
\left(1-u^{(i)}\right)^{-1}u^{(i)}=\psi^{(i)}_{12}\left(\psi^{(i)}_2\right)^{-1}, & \left(1-v^{(i)}\right)^{-1}v^{(i)}=\psi^{(i)}_{12}\left(\psi^{(i)}_1\right)^{-1},
\end{aligned}&&  i=1,\ldots,N.
\end{align}
In terms of the potential functions $\psi^{(i)},$ the set of equations (\ref{sec3:03}) is identically satisfied, while the set of equations (\ref{sec3:02}) becomes
\begin{align} \label{sec4:2:2}
\begin{multlined}[t]
\left(1+\psi^{(i)}_{12}\left(\psi^{(i)}_2\right)^{-1}\right)\left(1+\psi^{(i-1)}_2\left(\psi^{(i-1)}\right)^{-1}\right)=\left(1+\psi^{(i)}_{12}\left(\psi^{(i)}_1\right)^{-1}\right)
\\
\left(1+\psi^{(i-1)}_1\left(\psi^{(i-1)}\right)^{-1}\right),
\end{multlined}
\end{align}
$i=1,\ldots,N$ and constitute a hierarchy of difference systems in vertex variables.
\begin{prop} \label{prop4:2:1}
For the hierarchy of difference systems in vertex variables (\ref{sec4:2:2}) we have:
\begin{enumerate}
\item  it arises as the compatibility conditions of the Lax equation
\begin{align*}
L(\bff{ \psi}_{12},{\bff \psi}_2;\lambda)L({\bff \psi}_{2},{\bff \psi};\lambda)=L({\bff \psi}_{12},{\bff \psi}_1;\lambda)L({\bff \psi}_{1},{\bff \psi};\lambda),
\end{align*}
 associated with the strong Lax matrix
\begin{align*}
L={\scriptsize\begin{pmatrix}
                      \psi^{(N)}_1\left(\psi^{(N)}\right)^{-1}&0&\cdots&0&\lambda\,\left(1+\psi^{(N)}_1\left(\psi^{(N)}\right)^{-1}\right)\\
                      1+\psi^{(1)}_1\left(\psi^{(1)}\right)^{-1}&\psi^{(1)}_1\left(\psi^{(1)}\right)^{-1}& 0& \cdots &0 \\
                      0  &1+\psi^{(2)}_1\left(\psi^{(2)}\right)^{-1}&\ddots& {}  &\vdots\\
                      \vdots & &\ddots &\psi^{(N-2)}_1\left(\psi^{(N-2)}\right)^{-1} &0\\
                       0     & 0 &   & 1+\psi^{(N-1)}_1\left(\psi^{(N-1)}\right)^{-1}    &\psi^{(N-1)}_1\left(\psi^{(N-1)}\right)^{-1}
                     \end{pmatrix};}
\end{align*}
\item it is multidimensional consistent;
\item it  respects the {\it rhombic} symmetry
\begin{align*}
\tau :(\bff{\psi},{\bff \psi}_1,{\bff \psi}_2,{\bff \psi}_{12})\mapsto ({\bff \psi},{\bff \psi}_2,{\bff \psi}_1,{\bff \psi}_{12}),&&
\sigma : ({\bff \psi},{\bff \psi}_1,{\bff \psi}_2,{\bff \psi}_{12})\mapsto ({\bff \psi}_{12},{\bff \psi}_2,{\bff \psi}_1,{\bff \psi});
\end{align*}
\item it is an integrable hierarchy of difference systems in vertex variables.
\end{enumerate}
\end{prop}
\begin{proof}
The proof of this Proposition follows from the proof of the Proposition \ref{prop4:1}.
\end{proof}

Similarly as with  the hierarchy of difference systems (\ref{sec4:2}), the hierarchy of difference systems (\ref{sec4:2:2}), can be solved rationally only for the sets of variables  ${\bff \psi}$ and ${\bff \psi}_{12}$. If the centrality assumptions (\ref{sec3:5}) are imposed, that in terms of the potential functions ${\bff \psi}$ read
\begin{align} \label{sec4:2:5}
\prod_{l=1}^N\left(1+\psi^{(N-l+1)}\left(\psi^{(N-l+1)}_1\right)^{-1}\right)^{-1}=p\in C(\mathbb{A}^\times),\\ \label{sec4:55} \prod_{l=1}^N\left(1+\psi^{(N-l+1)}\left(\psi^{(N-l+1)}_2\right)^{-1}\right)^{-1}=q\in C(\mathbb{A}^\times),
\end{align}
 then (\ref{sec4:2:2}) can be solved  for any corner variable set, mimicking the  Abelian case.

  To recapitulate, hierarchy (\ref{sec4:2:2}) serves as an integrable hierarchy in vertex variables that respects the rhombic symmetry. If in addition the centrality assumptions (\ref{sec4:2:5}),(\ref{sec4:55}) are imposed, then one potential function can be eliminated and the resulting hierarchy consists of $N-1$ equations in $N-1$ potential functions. The resulting hierarchy  serves as the non-Abelian lattice-$(Q3)_0$ Gel'fand-Dikii  hierarchy, since its lowest member turns out to be $(Q3)_0$  (in the Abelian case goes back to \cite{nij-qui-cap2} and referred to as $(Q3)_0$ in \cite{ABS}),   as the example that follows suggest. Note that  the lattice-$(Q3)_0$ Gel'fand-Dikii  hierarchy in the Abelian case was firstly considered implicitly in \cite{Nijhoff:2011} and in \cite{Zhang:2012} it was explicitly derived just the second member of this hierarchy i.e. the  Boussinesq analogue of $(Q3)_0$.
 \begin{example}[$N=2$]
  The first member of the hierarchy (\ref{sec4:2:2}) ($N=2$)  reads:
 \begin{align}\label{q30}
 \begin{aligned}
 \begin{multlined}[t]
\left(\psi^{(1)}_{12}\left(\psi^{(1)}_2\right)^{-1}+1\right)\left(\psi^{(2)}_{2}\left(\psi^{(2)}\right)^{-1}+1\right)
=\left(\psi^{(1)}_{12}\left(\psi^{(1)}_1\right)^{-1}+1\right)\\
\left(\psi^{(2)}_{1}\left(\psi^{(2)}\right)^{-1}+1\right),\end{multlined}\\
 \begin{multlined}[t]
 \left(\psi^{(2)}_{12}\left(\psi^{(2)}_2\right)^{-1}+1\right)\left(\psi^{(1)}_{2}\left(\psi^{(1)}\right)^{-1}+1\right)
=\left(\psi^{(2)}_{12}\left(\psi^{(2)}_1\right)^{-1}+1\right)\\
\left(\psi^{(1)}_{1}\left(\psi^{(1)}\right)^{-1}+1\right).
\end{multlined}
 \end{aligned}
  \end{align}
  If we impose the centrality assumptions (\ref{sec4:2:5}),(\ref{sec4:55}) that now read
  \begin{align*}
  \begin{aligned}
  \left(1+\psi^{(2)}  \left(\psi^{(2)}_1\right)^{-1}\right)^{-1}  \left(1+\psi^{(1)}  \left(\psi^{(1)}_1\right)^{-1}\right)^{-1}=p,\\
     \left(1+\psi^{(2)}  \left(\psi^{(2)}_2\right)^{-1}\right)^{-1}  \left(1+\psi^{(1)}  \left(\psi^{(1)}_2\right)^{-1}\right)^{-1}=q,
\end{aligned}
  \end{align*}
  we can eliminate f.i. $\psi^{(2)}$ and its shifts from the first equation of (\ref{q30}) to obtain
  \begin{align*}
  \left(1+\psi_{12}\psi^{-1}_2\right)  \left(1-q\left(1+\psi \psi^{-1}_2\right)\right)=\left(1+\psi_{12}\psi^{-1}_1\right)  \left(1-p\left(1+\psi \psi^{-1}_1\right)\right),
  \end{align*}
  where for simplicity we have denoted $\psi:=\psi^{(1)},$ $\psi_2:=\psi^{(1)}_2,$ etc. Under the re-parametrization
  $  p\mapsto \frac{p^2}{p^2-1},$  $q\mapsto \frac{q^2}{q^2-1}$
   followed by the point transformation
     \begin{align*}
   \psi_{m+m_1,n+n_1}\mapsto (-1)^{m_1+n_1} p^{m_1} q^{n_1} \psi_{m+m_1,n+n_1},
   \end{align*}
    the equation above takes the form
  \begin{align} \label{Q30}
\left(\psi_2-p\psi_{12}\right)\left(\psi_2-\frac{q}{q^2-1}\left(q\psi_2-\psi\right)\right)^{-1}=
\left(\psi_1-q\psi_{12}\right)\left(\psi_1-\frac{p}{p^2-1}\left(p\psi_1-\psi\right)\right)^{-1}
  \end{align}
  that is exactly $(Q3)_0$ but extended in the non-Abelian domain.
 \end{example}

\begin{example}[$N>2$]
For $N>2$ we have the   lattice-$(Q3)_0$ Gel'fand-Dikii  hierarchy i.e. the set of equations (\ref{sec4:2:2}) that as we showed,  respect the rhombic symmetry.

From the centrality assumptions (\ref{sec4:2:5}),(\ref{sec4:55}), we obtain
\begin{align*}
1+\psi^{(N)}_1\left(\psi^{(N)}\right)^{-1}=\left(1-p\prod_{l=1}^{N-1}\left(1+\psi^{(l)}\left(\psi^{(l)}_1\right)^{-1}\right)\right)^{-1},\\
1+\psi^{(N)}_2\left(\psi^{(N)}\right)^{-1}=\left(1-q\prod_{l=1}^{N-1}\left(1+\psi^{(l)}\left(\psi^{(l)}_2\right)^{-1}\right)\right)^{-1}.
\end{align*}
Using the relations above, we can eliminate the potential $\psi^N$ and its shifts from (\ref{sec4:2:2}), to obtain the following  form of the   lattice-$(Q3)_0$ Gel'fand-Dikii  hierarchy.

\begin{multline*}
\left(1+\psi^{(1)}_{12}\left(\psi^{(1)}_2\right)^{-1}\right)\left(1-q\prod_{l=1}^{N-1}\left(1+\psi^{(l)}\left(\psi^{(l)}_2\right)^{-1}\right)\right)^{-1}\\
=\left(1+\psi^{(1)}_{12}\left(\psi^{(1)}_1\right)^{-1}\right)\left(1-p\prod_{l=1}^{N-1}\left(1+\psi^{(l)}\left(\psi^{(l)}_1\right)^{-1}\right)\right)^{-1},
\end{multline*}
\begin{multline*}
\left(1+\psi^{(i)}_{12}\left(\psi^{(i)}_2\right)^{-1}\right)\left(1+\psi^{(i-1)}_2\left(\psi^{(i-1)}\right)^{-1}\right)=\left(1+\psi^{(i)}_{12}\left(\psi^{(i)}_1\right)^{-1}\right)
\left(1+\psi^{(i-1)}_1\left(\psi^{(i-1)}\right)^{-1}\right),
\end{multline*}
\begin{multline*}
\left(1-p\prod_{l=1}^{N-1}\left(1+\psi^{(l)}_2\left(\psi^{(l)}_{12}\right)^{-1}\right)\right)^{-1}\left(1+\psi^{(N-1)}_2\left(\psi^{(N-1)}\right)^{-1}\right)\\
=\left(1-q\prod_{l=1}^{N-1}\left(1+\psi^{(l)}_1\left(\psi^{(l)}_{12}\right)^{-1}\right)\right)^{-1}\left(1+\psi^{(N-1)}_1\left(\psi^{(N-1)}\right)^{-1}\right),
\end{multline*}
$i=2,\ldots,N-1.$
\end{example}

\section{Conclusions}\label{Section:6}
In this article we introduced two families of hierarchies of non-Abelian    compatible maps. We have explicitly provided their associated Lax matrices,  we have shown that they serves as deformations of known hierarchies of maps \cite{Doliwa_2013,Doliwa_2014} and we have proven their  compatibility by explicitly providing their multidimensional compatibility formula. Furthermore, by imposing to these hierarchies certain compatible commutativity assumptions, we prove the quadrirationality of the latter and we provide explicitly the associated families of hierarchies of Yang-Baxter maps. Finally, we derived the corresponding integrable hierarchies of difference systems in non-commuting edge variables and the associated integrable difference hierarchies in vertex variables together with their Lax matrices. The lattice-modified Gel'fand-Dikki hierarchy and the lattice-NQC Gel'fand-Dikki hierarchy, both in non-commuting variables, together with the underlying integrable difference system in edge variables, were obtained.

The results of this article can be extended and/or generalized in various ways. We anticipate that by following \cite{Kouloukas:2011,Kassotakis:2019} we can obtain entwining hierarchies of maps associated with the ones presented here. Moreover,  by switching on the  deformation parameters (see the beginning of Section \ref{Section:3}), we can obtain degeneracies of the hierarchies $\mathcal{K}^{(i)},$ $i=1,2,$ that in turn will lead to degeneracies of the Gel'fand-Dikii hierarchies presented here, as well as their corresponding Yang-Baxter maps  which are expected to be related with the non-Abelian extension of the results in \cite{Kakei:2010}.

Furthermore $\mathcal{K}^{(1)},$ serves as a member of a  family of hierarchies which correspond to the following order $N\in\mathbb{N}$ Lax matrices
\begin{align*}
{\widehat L}^{(N,\kappa_1,\kappa_2)}:= D_L\, L^{(N,\kappa_1,\kappa_2)},
\end{align*}
where $D_L$ the diagonal  deformation matrix with entries  $(D_L)_{i,i}=\left(\alpha^{(i-\kappa_1)}-\beta x^{(i-\kappa_1)}\right)^{-1},$ $\beta,\alpha^j\in C(\mathbb{A}^\times)$ and
\begin{align*}
L^{(N,\kappa_1,\kappa_2)}:=\nabla^{(\kappa_2)}+\Delta^{(\kappa_2)}+\nabla^{(\kappa_1)}\,{\bf X}+\lambda\,\Delta^{(\kappa_1)}\,{\bf X},
\end{align*}
with the matrices $\nabla,\Delta,$ and ${\bf X}$ given as in Section \ref{Section:2} and $\kappa_1\neq\kappa_2\in\{1,\ldots,N\}.$ In this setting, $\mathcal{K}^{(1)}$ corresponds to the Lax matrix  ${\widehat L}^{(N,1,0)},$  ($\nabla^0+\Delta^0$  denotes the order $N$ identity matrix $I_N$).
Similarly $\mathcal{K}^{(0)},$ serves as a member of a dual family of hierarchies which correspond to the following order $N\in\mathbb{N}$ Lax matrices
\begin{align*}
{\widehat M}^{(N,\kappa_1,\kappa_2)}:=D_M\, M^{(N,\kappa_1,\kappa_2)},
\end{align*}
where $D_M$ the diagonal  deformation matrix with entries  $(D_M)_{i,i}=\left(\alpha-\beta^{(i)} x^{(i)}\right)^{-1},$ $\alpha, \beta^i\in C(\mathbb{A}^\times)$ and
\begin{align*}
M^{(N,\kappa_1,\kappa_2)}:=\nabla^{(\kappa_2)}\,{\bf X}+\Delta^{(\kappa_2)}\,{\bf X}+\nabla^{(\kappa_1)}+\lambda\,\Delta^{(\kappa_1)},
\end{align*}
where again $\kappa_1\neq\kappa_2\in\{1,\ldots,N\},$ and in this setting, $\mathcal{K}^{(0)}$ corresponds to the Lax matrix  ${\widehat M}^{(N,1,0)}.$ We postpone the study of the discrete spectral problems associated with the mentioned Lax matrices for a   future contribution.

\ack

\parbox{.135\textwidth}{\begin{tikzpicture}[scale=.03]
\fill[fill={rgb,255:red,0;green,51;blue,153}] (-27,-18) rectangle (27,18);
\pgfmathsetmacro\inr{tan(36)/cos(18)}
\foreach \i in {0,1,...,11} {
\begin{scope}[shift={(30*\i:12)}]
\fill[fill={rgb,255:red,255;green,204;blue,0}] (90:2)
\foreach \x in {0,1,...,4} { -- (90+72*\x:2) -- (126+72*\x:\inr) };
\end{scope}}
\end{tikzpicture}} 
{This research is part of the project No. 2022/45/P/ST1/03998  co-funded by the National Science Centre and the European Union Framework Programme
 for Research and Innovation Horizon 2020 under the Marie Sklodowska-Curie grant agreement No. 945339. For the purpose of Open Access, the author has applied a CC-BY public copyright licence to any Author Accepted Manuscript (AAM) version arising from this submission.}

\appendix
\section{Non-Abelian lattice-potential KdV equations} \label{app1}
A non-Abelian form of the lattice-potential KdV equation was firstly derived \cite{Field:2005}. Here we propose an alternative $D_4$ symmetric form of the non-Abelian lattice-potential KdV equation, as well as a non-Abelian form that respects the rhombic symmetry.

The following Lax matrix
 \begin{align*}
L({\bf x};\lambda):=\begin{pmatrix}
x^{(1)}&(x^{(1)}+x^{(2)})x^{(1)}-\lambda\\
1&x^{(1)}
\end{pmatrix},
 \end{align*}
 was introduced in \cite{Kassotakis:2:2021} and corresponds to a $3D-$compatible map that in the Abelian setting reduces to the companion map of the Adler map \cite{adler-1993}. The compatibility conditions $L({\bf u};\lambda)L({\bf y};\lambda)=L({\bf v};\lambda)L({\bf x};\lambda),$ are equivalent to the following set of equations
 \begin{align}\label{ap1}
 u^{(1)}+y^{(1)}=v^{(1)}+x^{(1)},\\ \label{ap2}
 \left(u^{(1)}+u^{(2)}\right)u^{(1)}-\left(y^{(1)}+y^{(2)}\right)y^{(1)}=\left(v^{(1)}+v^{(2)}\right)v^{(1)}-\left(x^{(1)}+x^{(2)}\right)x^{(1)}, \\ \label{ap3}
 \left(u^{(1)}+u^{(2)}\right)u^{(1)}+u^{(1)}y^{(1)}= \left(v^{(1)}+v^{(2)}\right)v^{(1)}+v^{(1)}x^{(1)},\\ \label{ap4}
  \begin{multlined}[t]\left(u^{(1)}+u^{(2)}\right)u^{(1)}y^{(1)}+u^{(1)}\left(y^{(1)}+y^{(2)}\right)y^{(1)}\\
  = \left(v^{(1)}+v^{(2)}\right)v^{(1)}x^{(1)}+v^{(1)}\left(x^{(1)}+x^{(2)}\right)x^{(1)},
\end{multlined}
 \end{align}
 that according to identification (\ref{notation1}), serve as a difference system in edge variables.
In that respect, equations (\ref{ap1}) and (\ref{ap2}), guarantee the existence of potential functions $\phi,\psi$ such that:
\begin{align}\label{potap0}
\begin{aligned}
x^{(1)}=\phi_1-\phi,&& y^{(1)}=\phi_2-\phi,\\
\left(x^{(1)}+x^{(2)}\right)x^{(1)}=\psi_1+\psi,&&\left(y^{(1)}+y^{(2)}\right) y^{(1)}=\psi_2+\psi.
\end{aligned}
\end{align}
In terms of these potential functions, (\ref{ap1}) and (\ref{ap2}) are identically satisfied while equations (\ref{ap3}) and (\ref{ap4}) respectively read
\begin{align*} 
  (\phi_{12}-\phi_2)(\phi_2-\phi)-(\phi_{12}-\phi_1)(\phi_1-\phi)=\psi_1-\psi_2,\\ 
    (\psi_{12}+\psi_2)(\phi_2-\phi)+(\phi_{12}-\phi_2)(\psi_2+\psi)=(\psi_{12}+\psi_1)(\phi_1-\phi)+(\phi_{12}-\phi_1)(\psi_1+\psi),
\end{align*}
and constitute a non-Abelian form of lattice-potential KdV equation. Note that this form of the lattice-potential KdV equation respects the rhombic symmetry, hence it is defined on a black-white lattice.

Furthermore, if the centrality assumptions $x^{(2)}x^{(1)}=-p,$ $y^{(2)}y^{(1)}=-q$ are assumed, relations (\ref{potap0}) give
\begin{align*} 
(\phi_1-\phi) (\phi_1-\phi)-p=\psi_1+\psi,&&(\phi_2-\phi) (\phi_2-\phi)-q=\psi_2+\psi,
\end{align*}
that serve as the B\"acklund transformation between the non-Abelian multiquadratic relation\footnote{In the Abelian case this multiquadratic relation was first obtained in \cite{Adler:2004} and serves as the nonlinear superposition principle of the B\"ackund transformation of the KdV equation} that the potential $\psi$ satisfies and
the $D_4$ symmetric form of the  non-Abelian  lattice-potential KdV equation that the potential function $\phi$ satisfies i.e.
\begin{align*} 
(\phi_1-\phi_2)(\phi_{12}-\phi)+(\phi_{12}-\phi)(\phi_{1}-\phi_2)=2(p-q).
\end{align*}


\end{document}